\pdfoutput=1

\documentclass[acmsmall,screen,authorversion]{acmart}
\authorsaddresses{}

\usepackage{booktabs}
\usepackage{subcaption}
\usepackage{xcolor}
\usepackage{xspace}
\usepackage{tikz}
\usetikzlibrary{shapes
  ,shapes.misc
  ,arrows.meta
  ,automata
  ,positioning
  ,arrows
  ,decorations.markings}
\usepackage[capitalize,nameinlink,noabbrev]{cleveref}
\usepackage{listings,lstautogobble}
\usepackage{wrapfig}
\usepackage{stmaryrd}
\usepackage{relsize}
\usepackage{extarrows}
\usepackage{tabularx}
\usepackage{mathpartir}
\usepackage{mathtools}
\usepackage{varwidth}
\usepackage{nicefrac,xfrac}
\usepackage{todonotes}

\newcommand{\mypara}[1]{\smallskip\noindent\emph{\textbf{#1}.\ }}

\newcounter{slogan}

\Crefname{slogan}{Slogan}{Slogans}

\newcounter{theme}
\newenvironment*{theme}[1][]
{\refstepcounter{theme}
  \begin{center}
    \begin{varwidth}{0.95\textwidth}
      \par
      \itshape
    }
    {
    \end{varwidth}
  \end{center}
}
\Crefname{theme}{Theme}{Themes}

\DeclareMathAlphabet{\mathcal}{OMS}{cmsy}{m}{n}

\newcommand{\tranz}{\mathsf{aut}}

\newcommand{\Asp}{\mathsf{Asp}}

\newcommand{\lang}{\textsc{Facet}}
\newcommand{\str}{\textsc{String}}
\newcommand{\mat}{\textsf{Match}}
\newcommand{\swit}{\textsf{Switch}}
\newcommand{\concat}{\mathsf{Concatenate}}
\newcommand{\constr}{\mathsf{ConstStr}}
\newcommand{\substr}{\mathsf{SubStr}}
\newcommand{\loope}{\mathsf{Loop}}
\newcommand{\cpos}{\mathsf{CPos}}
\newcommand{\poss}{\mathsf{Pos}}

\newcommand{\FO}{\mathsf{FO}}
\newcommand{\FOk}{\FO^{k}}

\newcommand{\Next}{\mathsf{X}}
\newcommand{\Until}{\mathsf{U}}
\newcommand{\Ex}{\mathsf{E}}
\newcommand{\wraparound}{\mathit{wrap}}

\newcommand{\MBoxx}{\ensuremath{\square}}
\newcommand{\MDiamm}{\ensuremath{\lozenge}}
\newcommand{\EG}{\ensuremath{\mathsf{EG}}}
\newcommand{\EU}{\ensuremath{\mathsf{EU}}}
\newcommand{\EX}{\ensuremath{\mathsf{EX}}}
\newcommand{\red}{a}
\newcommand{\blue}{v}

\newcommand{\pos}{\mathsf{pos}}
\newcommand{\subs}{\mathsf{subs}}

\newcommand{\lfp}{\mathsf{lfp}}
\newcommand{\gfp}{\mathsf{gfp}}

\newcommand{\arity}[1]{\mathsf{arity}(#1)}

\theoremstyle{definition}
\newtheorem*{definition}{Definition}
\newtheorem*{problem}{Problem}
\newtheorem*{lemma*}{Lemma}
\newtheorem*{theorem*}{Theorem}
\newtheorem{corollary}{Corollary}
\newtheorem{lemma}{Lemma}
\newtheorem{theorem}{Theorem}

\crefname{section}{\S}{\S}
\Crefname{section}{Section}{Sections}
\crefname{subsection}{Section}{Sections}
\Crefname{subsection}{Section}{Sections}
\crefname{equation}{}{}
\Crefname{equation}{Equation}{Equations}
\Crefname{lemma}{Lemma}{Lemmas}
\Crefname{theorem}{Theorem}{Theorems}
\Crefname{corollary}{Corollary}{Corollaries}
\Crefname{proposition}{Proposition}{Propositions}
\Crefname{corollary}{Corollary}{Corollaries}
\Crefname{algocf}{Algorithm}{Algorithms}
\crefname{enumi}{}{}
\crefname{figure}{Figure}{Figures}

\newcommand{\MSO}{\mathsf{MSO}}

\newcommand{\grammar}{\mathcal{G}}

\newcommand{\Nat}{\mathbb{N}}
\newcommand{\Rat}{\mathbb{Q}}
\newcommand{\Bool}{\mathbb{B}}
\newcommand{\Bb}{\mathcal{B}}
\newcommand{\Oo}{\mathcal{O}}
\newcommand{\Pos}{\mathit{P}}
\newcommand{\Neg}{\mathit{N}}
\newcommand{\var}{Var}
\newcommand{\varit}{\mathit{\var}}

\newcommand{\delvar}{\Delta\mathit{(\var)}}

\newcommand{\flip}{\mathit{flip}}

\newcommand{\grammarsep}{\,\, | \,\,}

\newlength\mylen

\newcommand{\update}[3]{#1[#2\mapsto#3]}
\newcommand{\tru}{\mathsf{True}}
\newcommand{\fals}{\mathsf{False}}
\newcommand{\aut}{\mathcal{A}}

\definecolor{zgray}{RGB}{63,63,63}
\definecolor{zyellow}{RGB}{220,203,155}
\definecolor{zgreen}{RGB}{143,175,159}
\definecolor{zcream}{RGB}{220,220,204}
\definecolor{zpink}{RGB}{220,163,163}
\definecolor{zblue}{RGB}{140,208,211}
\definecolor{assassinblue}{RGB}{110,150,210}
\definecolor{assassinred}{RGB}{210,106,115}
\definecolor{darkablue}{RGB}{30,70,120}
\definecolor{darkzgreen}{RGB}{60,100,80}

\newcommand{\Mm}{\mathcal{M}}

\newcommand{\Pp}{\mathcal{P}}
\newcommand{\Ll}{\mathcal{L}}

\newcommand{\singqt}[1]{\text{`}#1\text{'}}
\newcommand{\dblqt}[1]{\text{``}#1\text{''}}

\newcommand{\Prod}{\ensuremath{::=~}\xspace}

\colorlet{listing-comment}{gray}
\colorlet{operator-symbol}{yellow!45!black}
\lstdefinelanguage{Haskell}{
    language=Caml,
    morekeywords={match, with, case, of, any, all},
    morekeywords=[2]{False, True},
    keywordstyle=[2]\color{darkzgreen},
    morekeywords=[3]{int, bool, list, nat},
    keywordstyle=[3]\color{darkablue},
    literate=%
      {=}{{{\color{operator-symbol}$\mathtt{\coloneqq}$}}}1
      {LAM}{{{$\lambda$}}}1
      {cup}{{{$\cup$}}}1
      {sigma}{{{$\sigma$}}}1
      {conj}{{$\wedge$}}1
      {disj}{{$\vee$}}1
      {neg}{{$\neg$}}1
      {univ}{{$\forall x$}}1
      {exists}{{$\exists x$}}1
      {a_1}{{{$a_1$}}}1
      {w_i}{{{$w_i$}}}1
      {u_i}{{{u$_i$}}}1
      {v_i}{{{v$_i$}}}1
      {a_n}{{{$a_n$}}}1
      {e_1}{{{$e_1$}}}1
      {e_n}{{{$e_n$}}}1
      {e_m}{{{$e_m$}}}1
      {while}{{{while}}}1
      {<}{{{\color{operator-symbol}<}}}1
      {>}{{{\color{operator-symbol}>}}}1
      {:}{{{\color{operator-symbol}:}}}1
      {;}{{{\color{operator-symbol};}}}1
      {|}{{{\color{operator-symbol}|}}}1
      {[}{{{\color{operator-symbol}[}}}1
      {]}{{{\color{operator-symbol}]}}}1
      {\&}{{{\color{operator-symbol}\&}}}1
      {->}{{{\color{operator-symbol}$\rightarrow$}}}1
}

\lstdefinestyle{default}{
    basicstyle=\linespread{1.0}\ttfamily,
    columns=fullflexible,
    commentstyle=\sffamily\color{black!50!white},
    escapechar=\#,
    framexleftmargin=1ex,
    framexrightmargin=1ex,
    keepspaces=true,
    breakindent=0pt,
    keywordstyle=\color{darkablue},
    mathescape,
    showstringspaces=true,
    stepnumber=1,
    xleftmargin=0em,
}

\lstdefinestyle{smallstyle}{
    basicstyle=\footnotesize\ttfamily,
    columns=fullflexible,
    commentstyle=\sffamily\color{black!50!white},
    escapechar=\#,
    framexleftmargin=1ex,
    framexrightmargin=1ex,
    keepspaces=true,
    keywordstyle=\color{darkablue},
    mathescape,
    showstringspaces=true,
    stepnumber=1,
    xleftmargin=2.5em,
}

\newcommand\zlstinline{\let\par\endgraf\lstinline}

\newcommand{\linl}[1]{\lstinline!#1!}

\newcommand{\modal}{\hbox{\linl{Modal}}}
\newcommand{\ctl}{\hbox{\linl{CTL}}}
\newcommand{\fo}{\hbox{\linl{FO}}}
\newcommand{\route}{\mathit{root}}
\newcommand{\reg}{\hbox{\linl{Reg}}}
\newcommand{\no}{\mathit{dual}}
\newcommand{\cfg}{\hbox{\linl{CFG}}}

\newcommand{\Cnt}{\mathsf{Count}}
\newcommand{\ltl}{\hbox{\linl{LTL}}}
\newcommand{\rat}{\hbox{\linl{Rat}}}

\newcommand{\pat}{\mathit{pat}}

\newcommand{\lhs}{\mathit{lhs}}
\newcommand{\rhs}{\mathit{rhs}}
\newcommand{\lab}{\mathit{l}}
\newcommand{\child}{\mathit{c}}
\newcommand{\stay}{\mathit{stay}}
\newcommand{\dir}{\mathit{dir}}
\newcommand{\match}{\mathit{match}}
\newcommand{\clauses}{\mathit{C}}
\newcommand{\dual}{\mathit{dual}}
\newcommand{\subst}{\tau}
\newcommand{\eval}{\mathit{eval}}
\newcommand{\norm}{\mathit{norm}}
\newcommand{\nt}{\mathit{NT}}
\newcommand{\up}{\mathit{up}}
\newcommand{\find}{\mathit{find}}
\newcommand{\reset}{\mathit{reset}}
\newcommand{\topp}{\mathit{top}}
\newcommand{\eend}{\mathit{end}}
\newcommand{\term}{\mathit{term}}
\newcommand{\poly}{\mathit{poly}}
\newcommand{\place}{\mathit{place}}
\newcommand{\iseq}{\mathit{eq}}
\newcommand{\islt}{\mathit{lt}}
\newcommand{\pre}{\mathit{pre}}
\newcommand{\isgeq}{\mathit{geq}}
\newcommand{\isneq}{\mathit{neq}}

\newif\ifappendix\appendixtrue

\newcommand{\app}[2]{\ifappendix#1\else#2\fi}

\newcommand{\betterappref}[1]{\app{\Cref{#1}}{the extended version~\cite{extended-version}}}

\setcopyright{rightsretained}
\acmPrice{}
\acmDOI{10.1145/3586032}
\acmYear{2023}
\copyrightyear{2023}
\acmJournal{PACMPL}
\acmVolume{7}
\acmNumber{OOPSLA1}
\acmMonth{4}

\begin{document}

\title[Languages with Decidable Learning: A Meta-theorem]{Languages
  with Decidable Learning: A Meta-theorem}

\titlenote{Version with appendix.}

\author{Paul Krogmeier}
\email{paulmk2@illinois.edu}
\affiliation{ \department{Department of Computer Science}
  \institution{University of Illinois, Urbana-Champaign}
  \country{USA}
}

\author{P. Madhusudan}
\email{madhu@illinois.edu}
\affiliation{ \department{Department of Computer
    Science}
  \institution{University of Illinois,
    Urbana-Champaign}
  \country{USA}
}


\begin{abstract}
We study expression learning problems with syntactic restrictions
and introduce the class of \emph{finite-aspect checkable languages} to
characterize symbolic languages that admit decidable learning. The
semantics of such languages can be defined using a bounded amount of
auxiliary information that is independent of expression size but
depends on a fixed structure over which evaluation occurs. We
introduce a generic programming language for writing programs that
evaluate expression syntax trees, and we give a meta-theorem that
connects such programs for finite-aspect checkable languages to finite
tree automata, which allows us to derive new decidable learning
results and decision procedures for several expression learning
problems by writing programs in the programming language.
\end{abstract}



\maketitle
\renewcommand{\shortauthors}{P. Krogmeier, P. Madhusudan}


\section{Introduction}
\label{sec:introduction}

We undertake a foundational theoretical exploration of the exact
learning problem for \emph{symbolic languages} with rich
semantics. Learning symbolic concepts from data has myriad
applications, e.g., in
verification~\cite{madhu-qda,ice,data-driven-chc,learning-nn-controllers,neider2020}
and, in particular, invariant synthesis for distributed
protocols~\cite{koenig-taming-search-space, aiken-fo-sep,
  parno-21,distAI}, learning properties of
programs~\cite{inferring-rep-invariants, preconditions,
  learning-contracts}, explaining executions of distributed
protocols~\cite{neider-ltl-learning}, and synthesizing programs from
examples or specifications
\cite{mil-muggleton,evans-greffen-noisy,flashfill,sketch,sygusJournal,handa-rinard,fta-data-completion-scripts,Wang2017,flashmeta}.

In this paper, \emph{symbolic languages} are construed as sets of
expressions together with formal syntax and semantics. Languages
include logics, e.g., first-order and modal logics, programming
languages (functional or imperative), query languages like
\textsc{SQL}, or even languages whose expressions describe other kinds
of languages, e.g., regular expressions or context-free grammars. In
the \emph{exact learning problem} for a language $\Ll$, the goal is to
find an expression $e\in\Ll$ that is consistent with a given finite
set of (positively and negatively) labeled examples, which in this
setting are \emph{finite structures}. The expression $e$ should be
satisfied by all positive structures and not satisfied by any negative
ones. The semantic notion, i.e. \emph{satisfaction}, varies by
problem.

\mypara{Decidable Learning} The languages we study are complex enough
that polynomial time learning is seldom possible (even learning the
simplest Boolean formula that separates a labeled set of Boolean
assignments to variables is not possible in polynomial
time~\cite{kearns-vazirani}). Furthermore, assuming the semantics of
expressions over structures is computable (true in all languages we
consider), there is always a trivial algorithm that \emph{enumerates}
expressions, evaluates them over the given structures, and finds a
consistent expression if one exists. Given that enumeration can take
exponential time in the size of the smallest consistent expression to
terminate (if it terminates at all), and given any learning algorithm
may require exponential time, a meaningful theoretical analysis of
learning for such languages is hard. We hence consider \emph{decidable
  learning}, where hypothesis classes are infinite and learning
algorithms must terminate with a consistent expression if one exists
or report there is none. Note the trivial enumerator is not a decision
procedure when there are infinitely-many expressions, since it may not
be able to report an instance has no solution.

\mypara{Learning under Syntactic Restrictions} We require learning
algorithms to both \emph{accommodate syntactic restrictions} over the
language, i.e., restrictions to the hypothesis space, and to be able
to find \emph{small} expressions. These stipulations mitigate
overfitting. For instance, in the case of some logics, the set of
positively-labeled structures can be precisely captured using a single
formula that can be computed efficiently given the structures. Such a
solution is not interesting and is unlikely to generalize. In such
cases, learning also becomes trivially decidable: if there is any
consistent expression, then the highly specific one will be
consistent. Accommodating syntactic restrictions and requiring small
solutions circumvent these issues. Note also that syntactic
restrictions are a \emph{feature}--- one can always allow expressions
to be learned from the entire language.

\mypara{Learning in Finite-Variable Logics} Our work draws inspiration
from a recent result that showed classical logics, e.g., first-order
logic ($\FO$), have decidable learning when restricted to use
finitely-many variables~\cite{popl22}. The technique underlying this
result uses \emph{tree automata}. For each positively-labeled
(respectively negatively-labeled) structure, one builds a tree
automaton that reads expression syntax trees and accepts those that
are true (respectively false) in that structure. Such automata are
akin to \emph{version space algebras}~\cite{mitchell97}, and taking
their product (and a product with an automaton capturing the syntactic
restriction) results in an automaton that accepts \emph{all}
consistent expressions. Existence of solutions can be decided with
automata emptiness algorithms, which can be used to synthesize small
consistent expressions if they exist.

\mypara{Contributions} We show that the tree automata-theoretic
technique for learning extends much beyond finite-variable logics. We
prove decidable learning for a number of languages that are not
finite-variable logics, and we give a meta-theorem that streamlines
the task of proving decidable learning for new languages. It reduces
proofs to the problem of writing an interpreter for the semantics in a
particular \emph{programming language}, which we call $\lang$. Using
this meta-theorem, we exhibit a rich set of examples that have
decidable learning:
\begin{itemize}
\item \emph{Modal logic} over Kripke structures
  (\Cref{sec:motiv-exampl-modal,sec:modal-expr-separ})
\item \emph{Computation tree logic} over Kripke structures
  (\Cref{sec:modal-expr-separ})\footnote{Details can be found in \betterappref{sec:ctl}.}
\item \emph{Regular expressions} over finite words (\Cref{sec:regular-expressions})
\item \emph{Linear temporal logic} over periodic words (\Cref{sec:ltl})
\item \emph{Context-free grammars} over finite words (\Cref{sec:cfg})
\item \emph{First-order queries} over tuples of rationals numbers with
  order (\Cref{sec:rationals})
\item \emph{String transformations} from input-output examples
  (\Cref{sec:discussion}), similar to~\cite{flashfill}
\end{itemize}
In each of these settings, the learning problem is decidable under
syntactic restrictions expressed by a (tree) grammar, which is given
as an input along with the sample structures.

We emphasize our contributions are theoretical. The programming
language itself is a notational tool that identifies and abstracts a
pattern we observe many times in this work, that of \emph{programming
  with two-way tree automata} to prove decidable learning and derive
decision procedures. The learning algorithms we obtain have high
complexity; implementing a compiler from the programming language to
efficient decision procedures for learning will involve heuristics
that depend on specific problems. We note that learning problems for
several of the languages we study are of practical interest, with
previous work exploring algorithms for regular
expressions~\cite{jagadish-regexp-learning-08,fernau-regexp-learning-09},
linear temporal logic~\cite{neider-ltl-learning}, context-free
grammars~\cite{sakakibara-cfg-learning,langley-cfg-learning,vanlehn87-cfg-learning},
and string transformations in Microsoft Excel's Flash
Fill~\cite{flashfill}.

\smallskip \mypara{Meta-theorem
} Each of the results above follows from a \emph{meta-theorem} which
says, intuitively, that languages are decidably learnable as long as
expressions can be evaluated over any structure using a particular
kind of program. More precisely, we require a \emph{semantic
  evaluator} $P$ that, given a structure $M$ and expression $e$,
evaluates $e$ over $M$ by navigating up and down on the syntax tree
for $e$ using recursion. Furthermore, $P$ must rely only on a finite
set of \emph{semantic aspects of the structure} for memory during its
navigation over $e$. This set depends on $M$ but not on $e$. As long
as we can write such an evaluator, the meta-theorem guarantees
decidable learning for the language!

The notion of \emph{semantic aspects} is quite natural in the settings
we consider. In logic, the semantics of formulas $\varphi$ over a
structure $M$ is often presented by recursion on $\varphi$, with some
additional information. For instance, the satisfaction relation
$M, \gamma \models \varphi$ for $\FO$ uses an interpretation of
variables $\gamma$. For $\FO$ with a fixed set of $k$ variables, the
number of such $\gamma$ is \emph{bounded}; it depends on the structure
$M$ but not on the formula $\varphi$, and hence meets the
finite-aspect requirement we identify in this paper. Consequently, the
result on decidable learning for finite-variable $\FO$ is an immediate
corollary of our meta-theorem. In fact, all such results for
finite-variable logics~\cite{popl22} are obtained as
corollaries\footnote{The tree automata underlying each result for
  finite-variable logics can be easily translated to semantic
  evaluators of the kind we require. See 
  \betterappref{sec:fo}~ for such a semantic evaluator for $\FO$.}.

Semantic aspects are sometimes obvious from standard semantic
definitions and sometimes less so. In \emph{modal logic}, the standard
semantics is defined recursively in terms of the semantics for
subformulas at different \emph{nodes} in a Kripke structure, and
indeed, the aspects in this case are simply the nodes. For
\emph{computation tree logic} (CTL), standard semantics in the
literature would use the nodes as for modal logic but would go beyond
recursion in the structure of expressions and use recursive
definitions to give meaning to formulas (least and greatest
fixpoints~\cite{mcmillan-thesis}). In this case, the aspects include
the nodes of the structure as well as a \emph{counter} that encodes a
recursion budget for \emph{until} and \emph{globally} formulas to be
satisfied, with the counter value bounded by the number of nodes. The
standard semantics for a \emph{regular expression} $e$ defines the
language $L(e)$ recursively in the structure of $e$. For a given word
$w$, however, we specialize the semantics of membership in $L(e)$ to
$w$, using aspects that correspond to subwords of $w$ (e.g., a pair of
indices $(i,j)$ marking the left and right endpoints of a subword,
with $1\le i \leq j \le |w|$). This membership semantics involves a
finite number of aspects which is quadratic in the size of $w$ but
independent of $e$. Semantics of \emph{linear temporal logic} (LTL)
formulas over periodic words $uv^\omega$ can also be defined
(non-standardly) using a set of aspects corresponding to each position
of $u$ and $v$, again finite. The semantics for membership of a word
$w$ in the yield of a \emph{context-free grammar} (restricted to a
finite set of nonterminals) can again be written with aspects
corresponding to subwords, as for regular expressions. But in this
case it also requires navigating the tree representing the grammar
\emph{up and down} many times in order to parse $w$, which requires
keeping some extra memory. Standard semantics for \emph{first-order
  queries} over rational numbers with order would involve an
interpretation of variables as rational numbers, but this set is of
course infinite. It turns out that a finite set of aspects encoding
the \emph{ordering} of the variables is sufficient to define semantics
in this setting.

The meta-theorem is a powerful tool for establishing decidable
learning. We emphasize that its proof is technically quite simple---
programs that navigate trees using recursion can be translated to
\emph{two-way alternating tree automata}, which can be converted to
one-way tree automata to obtain decision procedures for learning. Our
technical contribution lies more in the formalization of the technique
in terms of a programming language for semantic evaluators, and
realizations in different settings with (nonstandard) semantic
definitions involving a finite set of aspects.

We use the meta-theorem for the well-known application of learning
string transfomations from examples in the context of spreadsheet
programs. The seminal work of Gulwani~\cite{flashfill} established
this problem as one of the first important applications of program
synthesis from examples. We consider the language for string programs
used in that work and argue that even a signficant extension of that
language admits decidable learning. As far as we know, decidable
learning for this well-studied problem was not known earlier.

\mypara{Organization} In \Cref{sec:motiv-exampl-modal} we explore
learning in modal logic to motivate the generic learning algorithm
based on tree automata and the notion of semantic aspects. We discuss
a semantic evaluator for modal formulas and abstract the main pattern
as a program. \Cref{sec:notation-background} gives some background on
tree automata. In \Cref{sec:language} we define the class of
finite-aspect checkable languages (languages that admit decidable
learning), formalize a programming language for writing semantic
evaluators, and give the meta-theorem connecting semantic evaluators
to tree
automata. \Cref{sec:regular-expressions,sec:ltl,sec:cfg,sec:rationals}
establish decidable learning for regular expressions, linear temporal
logic, context-free grammars, and first-order queries over rationals
with order. In \Cref{sec:discussion} we discuss decidable learning for
string transformations. We review related work in
\Cref{sec:related-work} and conclude in \Cref{sec:conclusion}.


\section{Motivating Problem: Learning Modal Logic Formulas}
\label{sec:motiv-exampl-modal}
In this section, we show how to derive learning algorithms from
semantic evaluators for propositional modal logic. We make the
observation that a specific kind of semantic evaluator corresponds to
a constructive proof of decidable learning. Specifically, the
evaluators must use an amount of memory \emph{bounded by the
  structure} over which expressions are evaluated but independent of
expression size, beyond that afforded by the syntax tree itself. To
prove decidable learning for a new language, it suffices to program
such an evaluator. We summarize this main theme as follows:

\begin{theme}
  Effective evaluation using state bounded by structures
  $\,\,\,\xRightarrow{\hspace*{0.5cm}}\,\,$ decidable learning
  \label{slogan1}
\end{theme}

We next introduce the learning problem for modal logic and explore
this theme by developing a suitable semantic evaluator for modal
formulas over Kripke structures.

\subsection{Separating Kripke Structures with Modal Logic Formulas}
\label{sec:separ-modal}
Consider the following problem, with an example illustrated in
\Cref{modal-picture}.

\begin{problem}[Modal Logic Separation]
  Given finite sets $P$ and $N$ of finite pointed Kripke structures
  over propositions $\Sigma$, and a grammar $\grammar$, synthesize a
  modal logic formula $\varphi\in L(\grammar)$ that is true for
  structures in $P$ (the \emph{positives}) and false for those in $N$
  (\emph{negatives}), or declare none exist.
\end{problem}

\noindent We review some basics of modal
logic~\cite{blackburn_rijke_venema_2001}. The following grammar
defines the set of modal logic formulas over a finite set of
propositions $\Sigma$.
\begin{align*}
  \varphi\Coloneqq a\in\Sigma \grammarsep \varphi\wedge\varphi'
  \grammarsep \varphi\vee\varphi' \grammarsep \neg\varphi \grammarsep
  \MBoxx\varphi \grammarsep \MDiamm\varphi
\end{align*}

\tikzstyle{every node}=[draw=none, inner sep=0pt, minimum width=8pt, minimum height=8pt]

\begin{figure}[H]
  \hspace{0.1in}
    \begin{minipage}[c]{0.4\textwidth}
      \begin{tikzpicture}[thick,scale=0.65]

        \node[] (-1) at (0,0.85) {} ;
        \node[] (-2) at (4,0.85) {} ;

        \node[] (0) at (0,0) {$\{{\color{assassinred}a}\}$} ;
        \node[] (1) at (1,-1) {$\{{\color{assassinred}c}\}$} ;
        \node[] (2) at (0,-2) {$\{{\color{assassinred}v}\}$} ;
        \node[] (3) at (-1,-1) {$\{{\color{assassinred}c}\}$} ;

        \draw[fill=zgray] [-to] (-1) edge[zgray] (0) ;

        \draw[fill=zgray] [-stealth] (0) edge[zgray] (1) ;
        \draw[fill=zgray] [-stealth] (1) edge[zgray] (2) ;
        \draw[fill=zgray] [-stealth] (2) edge[zgray] (3) ;
        \draw[fill=zgray] [-stealth] (3) edge[zgray] (0) ;

        \draw[fill=zgray] [-stealth] (0) edge[zgray] (2) ;
        \draw[fill=zgray] [-stealth] (2) edge[zgray] (0) ;

        \node[] (4) at (4,0) {$\{{\color{assassinred}c}\}$} ;
        \node[] (5) at (5,-1) {$\{{\color{assassinred}c}\}$} ;
        \node[] (6) at (4.5,-2) {$\{{\color{assassinred}a}\}$} ;
        \node[] (7) at (5.5,-2) {$\{{\color{assassinred}c}\}$} ;
        \node[] (8) at (3,-1) {$\{{\color{assassinred}c}\}$} ;
        \node[] (9) at (3.5,-2) {$\{{\color{assassinred}c}\}$} ;
        \node[] (10) at (2.5,-2) {$\{{\color{assassinred}v}\}$} ;

        \draw[fill=zgray] [-to] (-2) edge[zgray] (4) ;

        \draw[fill=zgray] [-stealth] (4) edge[zgray] (5) ;
        \draw[fill=zgray] [-stealth] (4) edge[zgray] (8) ;
        \draw[fill=zgray] [-stealth] (5) edge[zgray] (6) ;
        \draw[fill=zgray] [-stealth] (5) edge[zgray] (7) ;
        \draw[fill=zgray] [-stealth] (8) edge[zgray] (9) ;
        \draw[fill=zgray] [-stealth] (8) edge[zgray] (10) ;
      \end{tikzpicture}
    \end{minipage}%
    \begin{minipage}[c]{0.15\linewidth}
      \begin{tikzpicture}[thick,scale=0.7]
        \node[draw=none,fill=none] (phi) at (0,0) {$\MBoxx(\MDiamm(\red\vee \blue))$} ;
        \node[draw=none,fill=none] (+) at (-0.1,1.4) {{\color{assassinred}$+$}} ;
        \node[draw=none,fill=none] (-) at (0.9,1.4) {{\color{assassinblue}$-$}} ;
        \draw [thick,dash pattern={on 2pt off 2pt on 2pt off 2pt}]
        (-0.4,-1.6) -- (-0.1,-0.4) ;
        \draw [thick,dash pattern={on 2pt off 2pt on 2pt off 2pt}]
        (0.1,0.4) -- (0.4,1.6) ;
      \end{tikzpicture}
    \end{minipage}
    \begin{minipage}[c]{0.3\textwidth}
      \centering
      \begin{tikzpicture}[thick,scale=0.65]
        \node[] (-1) at (1.25,0.85) {} ;
        \node[] (-2) at (3.75,0.85) {} ;

        \node (cab) at (1.25,0) {$\{{\color{assassinblue}a}\}$} ;
        \node (cab2) at (1.25,-1.5) {$\{{\color{assassinblue}v}\}$} ;
        \draw[fill=zgray] [-stealth] (cab) edge[zgray] (cab2) ;

        \draw[fill=zgray] [-to] (-1) edge[zgray] (cab) ;

        \node (0) at (3.75,0) {$\{{\color{assassinblue}v}\}$} ;
        \node (1) at (4.5,-1.5) {$\{{\color{assassinblue}a}\}$} ;
        \node (2) at (3,-1.5) {$\{{\color{assassinblue}c}\}$} ;

        \draw[fill=zgray] [-to] (-2) edge[zgray] (0) ;

        \draw[fill=zgray] [-stealth] (0) edge[zgray] (1) ;
        \draw[fill=zgray] [-stealth] (1) edge[zgray] (2) ;
        \draw[fill=zgray] [-stealth] (2) edge[zgray] (0) ;
      \end{tikzpicture}
    \end{minipage}

    \caption{The modal logic formula
      $\varphi =\MBoxx(\MDiamm (\red\vee\blue))$ over
      $\Sigma=\{a,c,v\}$ is true for the two Kripke structures on the
      left and false for the two on the right. Starting nodes $s$ are
      on top with incoming arrows. 
    }
  \label{modal-picture}
\end{figure}

The standard semantics of modal logic is reproduced below. Formulas
are interpreted against (in our case finite) pointed Kripke structures
$G=(W,s,E,P)$, where $W$ is a set of nodes (or \emph{worlds}), $E$ is
a binary 
neighbor 
relation on $W$, and $P : W\rightarrow \Pp(\Sigma)$ is a function that
labels each node by the set of all atomic propositions that hold
there. A formula $\varphi$ is true in $G=(W,s,E,P)$, written
$G\models\varphi$, if it is true starting from $s$, written
$G,s\models\varphi$, with the latter notion defined as follows.

\vspace{0.1in}
\begin{minipage}[t]{0.6\textwidth}
\centering
\begin{tabularx}{0.6\textwidth}{l l l l l}
  $G,w$ & $\models$ & $a\in \Sigma$ \qquad\qquad & if &$a\in P(w)$ \\
  $G,w$ & $\models$ & $\neg\varphi$ \qquad\qquad & if &
  $G,w\not\models\varphi$ \\
  $G,w$ & $\models$ & $\varphi\wedge\varphi'$ \qquad\qquad & if &
  $G,w\models \varphi$ \text{ and } $G,w\models \varphi'$ \\
  $G,w$ & $\models$ & $\varphi\vee\varphi'$ \qquad\qquad & if &
  $G,w\models \varphi$ \text{ or } $G,w\models \varphi'$ \\
  $G,w$ & $\models$ & $\MBoxx \varphi$ \qquad\qquad & if &
  $G,w'\models \varphi$ \text{ for all} $w'$ \text{such that } $E(w,w')$\\
  $G,w$ & $\models$ & $\MDiamm \varphi$ \qquad\qquad & if &
  $G,w'\models \varphi$ \text{ for some} $w'$ \text{such that } $E(w,w')$
\end{tabularx}
\end{minipage}
\vspace{0.1in}

Observe that there are \emph{infinitely-many} inequivalent modal
formulas. Indeed, the sequence
\begin{align*}
\MDiamm a,\, \MDiamm (\MDiamm a),\, \MDiamm (\MDiamm (\MDiamm a)),\, ...
\end{align*}
defines an infinite set $(\varphi_i)_{i\in\Nat^{\mathbb{+}}}$ of
inequivalent formulas. For $i \in \Nat^{\mathbb{+}}$, a finite graph
consisting of a single directed path of length $i-1$ makes the formula
$\varphi_i$ false while making all $\varphi_j$ true for $j < i$. Thus
the search space of modal formulas is infinite, and so we cannot
resort to enumeration for decidable learning in modal logic.

We advocate an automata-theoretic technique for learning problems,
which is inspired by recent work on learning formulas in
finite-variable logics~\cite{popl22}:
\begin{enumerate}
\item Encode language expressions as syntax trees over a finite
  alphabet.
\item\label{aut-bullet} For each structure $p \in \Pos$ (respectively, $n\in\Neg$),
  construct a tree automaton accepting the syntax trees for
  expressions $e$ such that $p \models e$ (respectively,
  $n\not\models e$).
\item\label{aut-bullet2} Construct a tree automaton accepting the intersection of
  languages for automata from (\Cref{aut-bullet}), which accepts all
  expressions consistent with the examples.
\item Run an emptiness checking algorithm for the automaton from
  (\Cref{aut-bullet2}) to synthesize a (small) expression, or, if the
  language is empty, report \emph{unrealizable}.
\end{enumerate}

\noindent The procedure above adapts easily to learning with grammar
restrictions. Given a regular tree grammar $\grammar$, we can
construct a tree automaton accepting precisely the expressions allowed
by $\grammar$ and take its product with the automaton from
(\Cref{aut-bullet2}) before checking emptiness.

The crucial observation we make is that in order to apply this generic
procedure to learning problems for new languages, we need only
implement an \emph{evaluator\,} for the semantics of the language. For
any \emph{fixed structure $M$}, the evaluator checks whether
$M\models e$ for an input expression $e$, where $\dblqt{\models}$ is a
problem-specific semantic relationship. Using the evaluator, we can
compute the tree automaton for each given positive and negative
structure and proceed with the algorithm above.

We can view these semantic evaluators as \emph{programs} whose state
depends on the mathematical structure over which evaluation occurs but
depends only to a very small degree on the size of the expression
itself. The key to finding these programs is to consider the question
of how to interpret arbitrary input expressions from the language
(presented as syntax trees) against an arbitrary, but \emph{fixed},
structure. We invite the reader in the remainder of the section to
na{\"i}vely explore how to write a program that evaluates an input
modal logic formula $\varphi$ against a fixed Kripke structure by
traversing the syntax tree of $\varphi$.

\subsection{Evaluating Modal Formulas on Fixed Kripke Structures}
\label{sec:eval-modal-form}
We want a procedure for evaluating any formula $\varphi$ of modal
logic against a fixed Kripke structure $G=(W,s,E,P)$, where
\emph{evaluate} means \emph{verify} that $G\models \varphi$. The
evaluator hence is designed for any particular $G$ and takes the syntax tree of $\varphi$ as input.

Imagine we want to evaluate the formula
$\varphi=\MBoxx(\MDiamm(\red\vee \blue))$ from \Cref{modal-picture}
over the rightmost (tree-shaped) positive structure $G$. In
particular, we want to check whether $G,s\models \varphi$ holds by
traversing the syntax tree for $\varphi$ (displayed on the right
below) from the top down. Suppose $n$ is a pointer into the syntax
tree of $\varphi$, with $n$ initially pointing to the root. We first
read the symbol $\singqt{\MBoxx}$, and we recognize that
$G,s\models \varphi$ holds exactly when the subformula
$\MDiamm(\red\vee\blue)$ holds at each of the two children of $s$ in
$G$.
\begin{wrapfigure}{r}{0.13\textwidth}
  \vspace{0.1in}
  \centering
  \tikzstyle{every node}=[draw=none, inner sep=1.5pt, minimum width=0pt]
  \begin{tikzpicture}[scale=1,draw=none,thick,minimum width=0pt,inner sep=1.5pt]
    \node (box) at (0,0) {$\MBoxx$} ;
    \node (diam) at (0,-0.6) {$\MDiamm$} ;
    \node (or) at (0,-1.2) {$\vee$} ;
    \node (red) at (-0.5, -1.5) {$\red$} ;
    \node (blue) at (0.5, -1.5) {$\blue$} ;
    \draw [-] (box) edge[black] (diam) ;
    \draw [-] (diam) edge[black] (or) ;
    \draw [-] (or) edge[black] (red) ;
    \draw [-] (or) edge[black] (blue) ;
  \end{tikzpicture}
  \newline
\end{wrapfigure}
Let $w_1$ and $w_2$ stand for the children of $s$, and let
$\child_i(n)$ stand for the $i^{\mathit{th}}$ child of the syntax tree
pointed to by $n$. We now should recursively check whether
$G,w_1\models \MDiamm(\red\vee\blue)$ and
$G,w_2\models \MDiamm(\red\vee\blue)$ hold. To do this, we move
\emph{down} in the syntax tree by setting
$n\coloneqq \child_1(n) = \MDiamm(\red\vee\blue)$. We then need to
check $G,w_1'\models \red\vee\blue$ holds, where $w_1'$ is either the
left or right child of $w_1$ in $G$ (and likewise for $w_2$). Suppose
we \emph{nondeterministically guess} that
$G,w_1'\models \red\vee\blue$ holds with $w_1'$ being the left child
of $w_1$. We move \emph{down} once more by setting
$n\coloneqq \child_1(n) = \red\vee\blue$ and we \emph{verify} the
guess by checking $G,w_1'\models \red\vee\blue$, which plays out in a
similar way, with the traversal terminating and returning true because
$G,w_1'\models \blue$ holds, since $\blue\in P(w_1')$.

Note that the steps described above work for arbitrarily large
$\varphi$; indeed, each next step is determined by the current symbol
of the syntax tree and by some state, namely, the set of nodes $W$,
that depends entirely on $G$. Since $G$ is finite, so too is this
set. Observe also that the traversal required some computable
functions specific to Kripke structures. For example, we needed to
compute the assignment $P : V\rightarrow\Pp(\Sigma)$, membership for
$\Pp(\Sigma)$, and the set of $E$-neighbors of a given node.

\subsection{A Program for Evaluating Modal Formulas}
\label{sec:progr-eval-modal}
We conclude the example by writing a \emph{program} which captures our
traversal of $\varphi$ and the computation of whether
$G\models \varphi$. The program takes as inputs the structure $G$,
some auxiliary state $w$, and a pointer $n$ that initially points to
the root of the syntax tree for a modal formula.

The program $\modal$ is shown in \Cref{fig:modal-clause1}. 
We discuss formal semantics for such programs in \Cref{sec:language}.
Intuitively, the program implements the traversal sketched earlier by
first matching against the symbol $n.\lab$ which labels the current
node of the syntax tree. Depending on the symbol, it can then either
terminate by computing a Boolean function as its final answer (e.g.
\dblqt{$x\in P(w)$}) or it can combine the results of recursive calls
at nearby nodes on the syntax tree. It uses \linl{all} and \linl{any}
to represent finite conjunctions and disjunctions, and it uses a few
problem-specific computable functions, which we categorize as either
\emph{Boolean} functions or \emph{state} functions. The only Boolean
function in this case is for atomic propositions, i.e.
\dblqt{$z\in P(w)$}, and the only state function is for computing the
neighborhood of a given node in $G$, i.e.
\dblqt{$\{y\in G \,:\, E(w,y)\}$}. Negation in
\Cref{fig:modal-clause1} is handled by evaluating the negated
subformula in a dual state $\no(w)$, in which each part of the program
is interpreted as its dual, e.g., \linl{and} becomes \linl{or},
\emph{etc}. We return to these details in \Cref{sec:language}.

\begin{figure}
  \centering
  \begin{minipage}[c]{0.85\linewidth}
\begin{lstlisting}
  Modal($G$, $w$, $n$) = match $n.\lab$ with
      $\wedge$  -> Modal($G$, $w$, $n.\child_1$) and Modal($G$, $w$, $n.\child_2$)
      $\vee$  -> Modal($G$, $w$, $n.\child_1$) or Modal($G$, $w$, $n.\child_2$)
      $\neg$  -> Modal($G$, $\no(w)$, $n.\child_1$)
      $\MBoxx$  -> all (LAM$z$. Modal($G$, $z$, $n.\child_1$)) $\{y\in G \,:\, E(w,y)\}$
      $\MDiamm$  -> any (LAM$z$. Modal($G$, $z$, $n.\child_1$)) $\{y\in G \,:\, E(w,y)\}$
      $x$  -> $x\in P(w)$
\end{lstlisting}
  \end{minipage}
  \caption{$\modal$ evaluates modal formula $\varphi$ pointed to by
    $n$ against Kripke structure $G$ and checks $G\models\varphi$.}
  \label{fig:modal-clause1}
\end{figure}

Recall the theme from the beginning of this section:

\begin{theme}
  Effective evaluation using state bounded by structures
  $\,\,\,\xRightarrow{\hspace*{0.5cm}}\,\,$ decidable learning
  \label{slogan1}
\end{theme}

\noindent The evaluator for modal formulas uses auxiliary states that
depend on the number of nodes in the Kripke structure, and \emph{not}
on the size of the syntax tree. Strictly speaking, it accesses the
syntax tree using a pointer, and hence involves some minimal amount of
memory that depends on expression size, but this is the only such
dependence.

As we have just observed, effective evaluation of this sort is
possible for modal logic on finite Kripke structures, and programs
witnessing this fact like the one in \Cref{fig:modal-clause1} imply
decision procedures for learning. In the remainder of the paper, we
define a class of languages with decidable learning and formalize the
programming language for evaluators as well as a meta-theorem which
reduces proofs of decidable learning to the task of programming
evaluators. We write such programs to obtain results for several other
learning problems.


\section{Preliminaries}
\label{sec:notation-background}

Here we review some background on syntax trees, tree grammars, and
tree automata.

\subsection{Syntax Trees and Tree Grammars}
\label{sec:terms}

For each symbolic language in this paper we use a \emph{ranked
  alphabet} to form expression syntax trees. A ranked alphabet
$\Delta$ is a set of symbols $s$ equipped with a function
$\arity{s}\in\Nat$. For example, the ranked alphabet for modal
formulas over $\Sigma$ has $\arity{\MDiamm}=1$, $\arity{\wedge}=2$,
and $\arity{a}=0$ for each $a\in\Sigma$. We write $T_\Delta$ for the
set of $\Delta$-terms, or ($\Delta$-)\emph{syntax trees}, which is the
smallest set containing symbols of arity $0$ from $\Delta$ and closed
under forming new terms with symbols of greater arities. We write
$T_\Delta(X)$ for the set of $\Delta$-terms constructed with a fresh
set of nullary symbols $X$.

We use regular tree grammars to express syntax restrictions for
learning problems.  A \emph{regular tree grammar} is a tuple
$\grammar = (\nt, \Delta, S, P)$ consisting of a finite set of
nonterminals $\nt$, ranked alphabet $\Delta$, starting nonterminal
$S\in \nt$, and productions $P$. Each production has the form
``$A \rightarrow t$'', with $A\in \nt$ and $t\in
T_\Delta(\nt)$. 
We associate with the productions $P$ a reflexive and transitive
rewrite relation $\rightarrow_P^*$ on terms $T_\Delta(\nt)$
, and the language $L(\grammar)$ is the set
$\left\{t\in T_\Delta(\emptyset)\,\mid\, S\rightarrow_P^*
  t\right\}$. See~\cite{tata} for details.

\subsection{Tree Automata}
\label{sec:tree-automata}

Tree automata are finite state machines that operate on trees. We will
use tree automata that operate on finite trees, formally \emph{terms},
as described in \Cref{sec:terms}. Such automata are tuples
$\aut = (Q,\Delta,q_i,\delta,F)$ consisting of a finite set of states
$Q$, ranked alphabet $\Delta$, initial state $q_i\in Q$, transition
function $\delta$, and acceptance condition $F$. An automaton accepts
a tree $t\in T_\Delta$ if it has an \emph{accepting run} over $t$.
The notions of \emph{run} and \emph{accepting run} can vary.

In this work we use a convenient, though no more expressive, variant
of tree automata called an \emph{alternating two-way tree
  automaton}. Such automata walk up and down on their input tree and
branch using alternation to send copies of the automaton in updated
states to nearby nodes of the tree. We will only use
\emph{reachability} acceptance conditions in this paper, where
$F\subseteq Q$, and a tree is accepted if along every trajectory of
the automaton during its walk over the tree, it reaches a state in
$F$. We omit the formal definition of \emph{runs} for these automata,
which is entirely standard, though complicated, and unnecessary for
understanding our results.

For a symbol $s\in \Delta$ with $\arity{s}=k$ and state $q\in Q$, the
available transitions for a two-way alternating tree automaton are
described by a Boolean formula 
\begin{align*}
  \delta(q,s) \in \mathcal{B}^{\mathtt{+}}(Q\times\{-1,0,\ldots,k\}),
\end{align*}
where $\mathcal{B}^{\mathtt{+}}(X)$ means the set of positive Boolean
formulas over variables from a set $X$. 
Each variable $(q,m)$ represents a new state $q$ and direction $m$ to
take at a particular node in the tree, with $m=-1$ being a move
\emph{up} to the parent of the current node, $m=0$ meaning to
\emph{stay} at the current node, and the other numbers being moves
down into one of $k$ children. A subset of $Q\times\{-1,0,\ldots, k\}$
corresponds to a Boolean assignment, and the automaton can proceed
according to any assignment that satisfies the current transition
formula. For example, if the automaton reads symbol $h$ in state $q$,
the transition
\begin{align*}
  \delta(q,h)= ((q_1,1)\wedge(q_2,1))\vee
  ((q_1,2)\wedge(q_2,0)\wedge(q_1,-1))
\end{align*}
would allow either of the following: (1) continuing in states $q_1$
and $q_2$, each starting from the leftmost child, or (2) continuing
from $q_1$ in the second child from left, from $q_2$ at the current
node, and from $q_1$ in the parent.

Two-way alternating tree automata can be converted to one-way
nondeterministic tree automata with an exponential increase in
states~\cite{two-way-vardi,cachat-two-way}, and so they inherit
closure properties and standard decision procedures. In particular,
the emptiness problem can be solved in exponential time and a small
tree in the language can be synthesized in the same amount of time
when nonempty.  See~\cite{tata} for details.


\section{Meta-Theorem for Decidable Learning}
\label{sec:language}

In this section we define a rich class of languages for which
decidable learning is possible. We then develop a meta-theorem which
reduces proofs of decidable learning to writing semantic evaluators in
a programming language, which we call $\lang$\footnote{$\lang$ stands
  for \underline{f}inite \underline{a}spect \underline{c}heckers of
  \underline{e}xpression \underline{t}rees.}. The decision procedures
involve an effective translation of $\lang$ programs into two-way
alternating tree automata that read syntax trees. After defining the
class (\Cref{sec:decid-learn-via}), we discuss the syntax and
semantics of $\lang$ (\Cref{sec:syntax-semantics}), followed by the
meta-theorem (\Cref{sec:main-lemma}), which says that any language
whose semantics can be evaluated by a $\lang$ program has decidable
learning. We then apply this theorem to show decidable learning for
modal logic and computation tree logic (\Cref{sec:modal-expr-separ}).

\subsection{A Class of Languages with Decidable Learning}
\label{sec:decid-learn-via}
There is a surprisingly rich set of languages that have decidable
learning via essentially one generic decision procedure, which we
describe here. For our purposes, a \emph{language} consists of a set
of expressions $\Ll$\footnote{We often abuse notation and do not
  distinguish between expressions $e\in\Ll$ and the syntax trees for
  $e$.}, a class of finitely-representable structures $\Mm$ over the
same signature, and a semantic function that interprets a structure
and an expression in some domain $D$, written with a turnstile as
$(\_\models\_) : \Mm\times\Ll\rightarrow D$. Sometimes we just use
$\Ll$ to refer to such a symbolic language.

The decision procedure relies on building a tree automaton that
accepts the set of all (syntax trees for) expressions $e\in\Ll$ that
are consistent with a given example. In a supervised learning
scenario, with $D=\Bool\coloneq \{\tru,\fals\}$ and examples modeled
as pairs $(M,b)\in \Mm\times \Bool$, one builds a tree automaton
$\aut(M,b)$ such that
$L(\aut(M,b))=\{e\in\Ll \,:\, M \models e = b\}$. For a finite set of
examples $E=(M_i,b_i)_i$, we take the product
$\aut(\mathit{E}) = \bigwedge_i \aut(M_i,b_i)$. Given an automaton
$\aut(\grammar)$ that accepts syntax trees conforming to a tree
grammar $\grammar$, we construct the product
$\aut(E) \wedge \aut(\grammar)$ and run an emptiness algorithm on the
result to synthesize a tree in the language if one exists.

The crucial requirement above is to be able to build $\aut(M,b)$ for
any $M$, which is an automaton that acts as an evaluator for the
language over $M$. This is possible when the semantics of a language
is definable in terms of a \emph{finite} amount of auxiliary
information, which may depend (sometimes wildly) on the particular
structure $M$ but not on the expression size. We refer to such
auxiliary semantic information as \emph{semantic aspects}, or just
\emph{aspects}, and we call languages for which evaluators can be
implemented using tree automata \emph{finite-aspect
  checkable}\footnote{
    \emph{Checkable} refers to the \emph{model
    checking} problem for a logic, i.e. checking whether
  $M\models\varphi$ for a structure $M$ and formula $\varphi$.
}.

\begin{definition}[Finite-Aspect Checkable Language]
  A language $(\Mm,\Ll,\models)$ is \emph{finite-aspect checkable}
  (FAC) if for every $(M,d)\in\Mm\times D$ there is a tree automaton
  $\aut(M,d)$ over syntax trees for $\Ll$ such that
  $L(\aut(M,d)) = \{e\in \Ll \,:\, M\models e = d\}$, and the mapping
  $(M,d) \mapsto \aut(M,d)$ is computable.
\end{definition}
\noindent Note that \emph{all} FAC languages have decidable learning
by the generic algorithm described above.

FAC languages only require the automata to be computable given
$(M,d)\in\Mm\times D$, but all examples we have considered in fact
have small \emph{witnesses} for being FAC: the tree automata can be
described compactly by a \emph{program} that evaluates an input syntax
tree against an input structure. We next describe the programming
language $\lang$, which abstracts the common features of such
programs. Our meta-theorem relies on a simple procedure that takes a
program $P\in\lang$, a structure $M\in\Mm$, and a domain element
$d\in D$, and computes the tree automaton $\aut(M,d)$.

\subsection{Syntax and Semantics of $\lang$}
\label{sec:syntax-semantics}
We present the syntax and semantics of $\lang$ by way of example. We
omit many details that are not important for understanding later
sections and results; details, including formal semantics of $\lang$,
can be found in \betterappref{sec:app-language}.

Programs in $\lang$ are parameterized by a symbolic language
$(\Ll,\Mm,\models)$. A program $P$ takes as input a \emph{pointer}
into the syntax tree for an expression $e\in\Ll$ as well as a
structure $M\in\Mm$. The program navigates up and down on $e$ using a
set of pointers to move from children to parent and parent to children
in order to evaluate the semantics of $e$ over the structure $M$ and
verify that $M\models e = d$ for some $d\in D$. To write a $\lang$
program we first specify two things: (1) the symbolic language $\Ll$
over which the program is to operate and (2) the program's auxiliary
state, which corresponds to the semantic aspects of $\Ll$. Part (1)
involves specifying (1a) the syntax trees for $\Ll$ in terms of a
ranked alphabet $\Delta$ and (1b) the signature for structures $\Mm$,
which is a set of functions used to access the data for any given
$M\in\Mm$. The set of auxiliary states of part (2), which we denote by
$\Asp$, will typically be infinite. For a fixed structure $M\in\Mm$,
however, programs will only use a finite subset $\Asp(M)\subset\Asp$,
provided the symbolic language $\Ll$ is FAC, and so we will only need
to specify $\Asp(M)$ for an arbitrary fixed $M$.

For example, consider the program \linl{Modal} for modal logic in
\Cref{fig:modal-clause1}. Part (1): the symbolic language is modal
logic over finite pointed Kripke structures $G=(W,s,E,P)$ with
propositions $\Sigma$. We fix any straightforward representation of
formulas as syntax trees, and we use two Kripke structure-specific
functions for interpreting modal logic. The first computes the
neighborhood $\{y\in G\,:\, E(w,y)\}$ of a given node $w\in W$, and
the second computes whether a given proposition $x\in\Sigma$ is true
at a given node $w\in W$, i.e. whether $x\in P(w)$ holds. Part (2):
the states $\Asp(M)$ are the nodes of the Kripke structure, i.e. the
set $W$.

\subsubsection{Syntax}
\label{sec:syntax-1}
The formal syntax for $\lang$ programs is shown in \Cref{fig:lang}. A
program $P\in\lang$ consists of a set of \emph{clauses}, which we
denote by $\clauses(P)$, or just $\clauses$, each of which has the
form
\begin{center}
\linl{$P$($M$,\ $\sigma(z)$,\ $n$) = match\ $n.\lab$\ with\ $\ldots$}
\end{center}
The parameter $M$ is a mathematical structure (e.g. a Kripke
structure) and the parameter $n$ is a pointer into a syntax tree
(e.g. for a modal logic formula). The parameter $\sigma(z)$ is a
\emph{pattern} (e.g. a variable $w$ matching any node of a Kripke
structure). We treat both the ranked alphabet $\Delta$ and auxiliary
states $\Asp(M)$ as algebraic data types and allow $\lang$ programs to
pattern match over these using expressions from two sets of patterns,
alphabet patterns $\pat(\Delta)$ and state patterns $\pat(\Asp)$. For
example, for \linl{Modal} in \Cref{fig:modal-clause1}, the (trivial)
state pattern $w\in\pat(\Asp)$ will match any node of the input Kripke
structure, and the (trivial) alphabet pattern $x\in\pat(\Delta)$ will
match any of the modal logic propositions in $\Sigma$.

Each clause in $C$ has a single \linl{match} statement, consisting of
a list of \emph{cases}, each of the form \dblqt{\linl{$\alpha_i(z)$ ->
    $\,\,e$}}, with an alphabet pattern on the left and an expression
on the right. Expressions $e$ represent Boolean functions, possibly
involving results of recursive calls \dblqt{\linl{P($M$, $\sigma(z)$,
    $n.\dir$)}} that start in new states $\sigma(z)$ at nearby nodes
$n.\dir$ on the syntax tree. Compound expressions are built using
\linl{and}, \linl{or}, \linl{all}, \linl{any}, and \linl{if}. For
example, in \Cref{fig:modal-clause1}, the first two cases involve
recursive calls at the two children ($n.\child_1$ and $n.\child_2$) of
the current node, in the same state $w$, and return, respectively, the
conjunction and disjunction of the results.

The sets $S$ and $B$ in \Cref{fig:lang} categorize the signature
functions for structures $\Mm$ as one of two kinds.
\emph{\underline{S}tate} functions $g\in S$ are used to compute new
states for the program, e.g. computing the neighborhood of a given
node in a Kripke structure. These functions are used in \linl{any} and
\linl{all} expressions to bind parts of the structure $M$ to
variables. \emph{\underline{B}oolean} functions\footnote{We assume the
  set of Boolean functions $B$ is closed under complement.} $f\in B$
are used in \linl{if} expressions as well as for base cases in
recursion, e.g. computing membership in the set of true propositions
for each node of a Kripke structure. Note that for all the symbolic
languages in this paper, the signature functions are evidently
computable.

\begin{figure}
\begin{tabularx}{0.95\textwidth}{l c l}
  $\mathit{Prog}$ & \Prod{} & $\left\{\, \mathit{Clause}\,\ldots\, \mathit{Clause} \,\right\}$ \\
  $\mathit{Clause}$ & \Prod{} & \linl{P($M$,\ $\sigma(z)$,\ $n$) = match\ $n.\lab$ with\ $\mathit{Cases}$} \\
  $\mathit{Cases}$ & \Prod{} & \linl{$\alpha_1(z)$ ->\ $e_1$\ $\ldots$\ $\alpha_n(z)$ ->\ $e_n$}
\end{tabularx}

\begin{tabularx}{0.95\textwidth}{@{}r@{\ }c@{\ }c@{\ }l@{\hspace*{2em}}c@{\ \ \ }l@{\ \ \ }@{\hspace*{2em}}c@{\ \ \ }l}
\\
  $e$ &        &\Prod{} & \linl{True} & | & \linl{False} & | & {\lstinline!$f$($z$)!} \\
  &        & | & {\lstinline|$e_1$ and $\ e_2$|} & | & {\lstinline|$e_1$ or $\ e_2$|} & | & {\linl{P($M$,\ $\sigma(z)$,\ $n.\dir$)}} \\
  &        & | & {\lstinline|all ($\lambda x$.\ $e$) $\
    g$($z$)|} & | & {\lstinline|any ($\lambda x$.\ $e$) $\
    g$($z$)|} & | & {\lstinline|if $\ f(z)$ then $\ e_1$ else $\ e_2$|}
  \\
\end{tabularx}

\begin{tabularx}{0.97\textwidth}{c c c c c c}
  \\
  $\alpha(z) \in \pat(\Delta)$
  &
  &
  $\sigma(z) \in \pat(\Asp)$
  &
  &
  $f \in B, \,\, g \in S\quad$
  & $\dir \in \{\mathit{\up,\stay,\child_1,...,\child_k}\}$ \\
\end{tabularx}
\caption{Syntax for $\lang$ programs.  We use $x$ to denote a single
  variable and $z$ to denote a vector of variables.  }
\label{fig:lang}
\end{figure}

\subsubsection{Semantics}
\label{sec:semantics-1}

Given a structure $M$, state $\sigma\in\Asp(M)$, and pointer $n$, a
program operates by first determining the clause whose state pattern
$\sigma(z)$ matches $\sigma$. We consider only \emph{well-formed}
programs in which every state is matched by the state pattern of
precisely one clause
. After the unique matching clause is determined, the program matches
the symbol $n.\lab$ labeling the current node of the syntax tree
against the alphabet patterns. It then evaluates the expression on the
right side of the first matching case and returns a Boolean result,
either success or failure.

A formal semantics for $\lang$ can be found in
\betterappref{sec:app-language}. It defines a predicate $\Downarrow_p$ on program
configurations of the form $(M,n,\clauses(P),\sigma)$. The assertion
$(M, n, \clauses(P),\sigma) \Downarrow_p$ has the following meaning:
the program consisting of clauses $\clauses(P)$ terminates with
success (that is, it computes ``true'') on the syntax tree pointed to
by $n$, when started from state $\sigma$ and working over the
structure $M$.

Consider the operation of the program \linl{Modal} from
\Cref{fig:modal-clause1} over the tree-shaped positive Kripke
structure from \Cref{modal-picture} and the formula
$\varphi =\MBoxx(\MDiamm(a\vee v))$. Suppose $n$ is pointing at the
root of $\varphi$, with label $\MBoxx$, and the current state is
$s\in W$ (the top node of the Kripke structure). The state $s$ matches
the state pattern of the clause depicted in \Cref{fig:modal-clause1},
and the variable $w$ is bound to $s$. Next, the symbol $n.\lab=\MBoxx$
matches the fourth case of the \linl{match} statement, and the program
evaluates the \linl{all} expression. To do this, it uses a state
function to compute the neighborhood
$N(s) = \{y \in G \,:\, E(s,y)\}$. It then returns the conjunction of
results for recursive calls obtained by evaluating \linl{Modal($G$,
  $z$, $n.\child_1$)} with $z$ bound to the states of $N(s)$. This
involves two recursive calls with pointer $n.\child_1$ corresponding
to the formula $\MDiamm(a\vee v)$ in states corresponding to the two
neighbors of $s$. Each of these recursive calls involves evaluating
the \linl{any} expression for the $\MDiamm$ case, which in turn
involves evaluating the expression for the $\vee$ case. Finally, when
the program encounters either of the propositions $a,v\in\Sigma$ in
the syntax tree, the last case will match, and the program uses a
function to compute $x\in P(w)$, i.e. whether the proposition bound to
$x$ is true at the current state.

\subsection{Meta-Theorem}
\label{sec:main-lemma}
We now connect $\lang$ programs with automata to give our meta-theorem
for decidable learning. The proof relies on the following lemma, which
states that programs in $\lang$ that use finite auxiliary states for
any structure can be translated to two-way alternating tree automata.

\begin{lemma}
  Let $(\Ll,\Mm,\models)$ be a symbolic language and $\Delta$ an
  alphabet for $\Ll$. Let $P$ be a well-formed $\lang$ program over
  $(\Ll,\Mm,\models)$ with computable signature functions and with
  $\Asp(M)$ finite for every $M\in\Mm$. Then for every $M\in\Mm$ and
  $\sigma\in\Asp(M)$, we can compute a two-way alternating tree
  automaton $\aut(P,M,\sigma) = (Q, \Delta, q_i, \delta, F)$ such that
  for every $\,e\in\Ll$, we have $e\in L(\aut(P,M,\sigma))$ if and
  only if $(M, \route(e), \clauses(P),\sigma) \Downarrow_p$.
  \label{lemma:mainlemma-lemma}
\end{lemma}
\begin{proof}[Proof Sketch]
  The automaton states are $Q = \Asp(M) \sqcup\{q_\top,q_\bot\}$, with
  $q_\top$ and $q_\bot$ being absorbing states for accepting and
  rejecting upon termination of the program. The initial state is
  $q_i = \sigma$, and the transitions $\delta$ are obtained by a
  straightforward translation of expressions into Boolean formulas,
  detailed in 
  \betterappref{sec:translate}. The acceptance condition is reachability with
  $F = \{q_\top\}$. The correspondence between the language of the
  automaton and semantic proofs for $P$ is straightforward and
  essentially follows by construction.
\end{proof}

\begin{theorem}[Decidable Learning Meta-Theorem]
  Let $(\Ll,\Mm,\models)$ be a language. If there is a semantic
  evaluator, i.e., a well-formed program $P\in \lang$, such that for
  all $M\in\Mm$, $\Asp(M)$ is finite, and for all $e\in \Ll$ and
  $d\in D$ there is a $\sigma_d\in\Asp(M)$ for which
  $(M, \route(e), \clauses(P), \sigma_d)\Downarrow_p$ if and only if
  $M\models e = d$, then $(\Ll,\Mm,\models)$ has decidable learning.
  \label{lemma:mainlemma}
\end{theorem}
\begin{proof}
  Let $P$ be an evaluator for a language $(\Ll,\Mm,\models)$, let
  $(M_i,d_i)_i$ be a finite set of examples from $\Mm\times D$, and
  let $\grammar$ be a tree grammar for $\Ll$. Build the tree automata
  $\aut(P,M_i,\sigma_{d_i})$ from
  \Cref{lemma:mainlemma-lemma}. Construct the product of these
  automata and convert the result to a nondeterministic tree automaton
  $\aut$. Take the product of $\aut$ with a nondeterministic tree
  automaton for $\grammar$, and use an emptiness algorithm to
  synthesize an expression or decide there is none.
\end{proof}

\paragraph{Remark on Complexity.} The complexity of decision
procedures for learning can be read from the number of aspects used by
the evaluation program, since they correspond to automaton
states. Products of two-way tree automata obtained from
\Cref{lemma:mainlemma-lemma} are converted to one-way nondeterministic
tree automata with an exponential increase in states using known
algorithms~\cite{two-way-vardi,cachat-two-way}. This leads to decision
procedures with time complexity exponential in the number of aspects
as well as the number of examples. Provided that signature functions
are computable in time exponential in the size of structures (true for
all languages in this paper), we can use the following:

\begin{corollary}
  Decision procedures for learning obtained via \Cref{lemma:mainlemma}
  have time complexity exponential in the number of examples and the
  number of aspects, and linear in the size of the grammar.
  \label{corollary:complexity-lemma}
\end{corollary}

$\lang$ enables compact descriptions of two-way tree automata, and
thereby enables decision procedures for learning to be developed using
intuition from programming. If we can write a $\lang$ program to
interpret a language $\Ll$ using a fixed amount of auxiliary state for
any given structure, then the language is FAC and decision procedures
for learning and synthesis follow from results in automata theory.

\subsection{Decidable Learning for Modal Logic and Dual Clauses}
\label{sec:modal-expr-separ}
We finish this section with a decidable learning theorem for modal
logic by completing the $\modal$ program from
\Cref{sec:motiv-exampl-modal} and explaining \emph{dual} states and
programs, which are syntactic sugar useful for handling negation and
negative examples.

The grammar for modal logic formulas over propositions $\Sigma$ from
earlier has a straightforward ranked alphabet $\Delta$, with members
of $\Sigma$ having arity $0$. The class $\Mm$ consists of finite
pointed Kripke structures $G=(W,s,E,P)$. For a given
$G=(W,s,E,P)$ we have
\begin{align*}
  \Asp(G)=\{w,\no(w) \,:\, w\in W\},
\end{align*}
where $\no$ is a constructor for states related to negation. There are
two signature functions: a state function for computing neighborhoods
$\{y\in G \,:\, E(w,y)\}$ and a Boolean function for computing
membership in $P(w)$, for a given $w\in W$.  Along with the clause in
\Cref{fig:modal-clause1}, $\modal$ includes the \emph{dual} of that
clause, shown in \Cref{fig:modal-clause2}, which operates on states of
the form $\no(x)$. For any clause $c$ there is a simple translation to
produce its dual clause $\dual(c)$ as follows:

\vspace{0.1in}
\noindent \begin{minipage}[c]{\linewidth}
  \begin{tabularx}{\linewidth}{r l X l}
    $c\,\,\, =$ &$\text{\lstinline!$P$($M$,\ $\sigma(z)$,\ $n$)!}$ &
    $\text{\lstinline!= match\ $n.\lab$\ with\ $\alpha_1$ ->\ $e_1$!}$ & $\ \ ...\ \ \text{\lstinline!$\alpha_n$ ->\ $e_n$!}$ \\
    $\dual(c)\, =$ &$\text{\lstinline!$P$($M$,\ $\no(\sigma(z))$,\ $n$)!}$ & $\text{\lstinline!= match\ $n.\lab$\ with\ $\alpha_1$ ->\ $\dual(e_1)$!}$ & $\ \ ...\ \ \text{\lstinline!$\alpha_n$ ->\ $\dual(e_n)$!}$
  \end{tabularx}
\end{minipage}
\vspace{0.1in}

The expressions $\dual(e_i)$ are obtained by recursively swapping
\linl{True} with \linl{False}, \linl{and} with \linl{or}, \linl{all}
with \linl{any}, etc., as follows: \vspace{0.1in}

\hspace{-0.1in}
\begin{minipage}[l]{\linewidth}
  \begin{tabularx}{\linewidth}{l l}
    $\no($\linl{True}$) =\ $ \linl{False} & $\no($\linl{False}$) =\ $\linl{True} \\
    $\no(e$ \linl{and} $e') = \no(e)$ \linl{or} $\no(e')$ & $\no(e$ \linl{or}  $e') = \no(e)$ \linl{and} $\no(e')$ \\
    $\no($\linl{$P$($M$,\ $\sigma(z)$,\ $n.\dir$)}$) =\ $ \linl{$P$($M$,\ $\flip(\sigma(z))$,\ $n.\dir$)} & $\no(f(v)) = \neg f(v)$ \\
  \end{tabularx}
  \begin{tabularx}{\linewidth}{l}
    $\no($\linl{all (LAM$z$.\ $e$)\ $g(v)$}$) =\ $\linl{any (LAM$z$.\ $\no(e)$)\ $g(v)$} \\
    $\no($\linl{any (LAM$z$.\ $e$)\ $g(v)$}$) =\ $\linl{all (LAM$z$.\ $\no(e)$)\ $g(v)$} \\
    $\no($\linl{if\ $f(v)$\ then\ $e_1$ else\ $e_2$}$)$ = \linl{if\ $f(v)$\ then\ $\no(e_1)$ else\ $\no(e_2)$}
  \end{tabularx}
\end{minipage}

\vspace{0.1in}
\noindent where $\flip(\no(\sigma(z))) = \sigma(z)$ and otherwise
$\flip(\sigma(z)) = \no(\sigma(z))$. For the program \linl{Modal}, the
dual and non-dual clauses invoke each other whenever a negation
operator is encountered in the syntax tree. But dual clauses are
useful even if the language $\Ll$ has no negation operation, given we
may need to check that semantic relationships \emph{do not} hold for
negative structures. Also note that the dual transformation above is a
syntactic notion for $\lang$ programs, entirely independent of the
symbolic language $\Ll$. 
From now on we omit the dual clauses from our presentation.

\begin{figure}
  \centering
  \begin{minipage}[c]{0.9\linewidth}
\begin{lstlisting}
  Modal($G$, $\no(w)$, $n$) = match $n.\lab$ with
      $\wedge$  -> Modal($G$, $\no(w)$, $n.\child_1$) or Modal($G$, $\no(w)$, $n.\child_2$)
      $\vee$  -> Modal($G$, $\no(w)$, $n.\child_1$) and Modal($G$, $\no(w)$, $n.\child_2$)
      $\neg$  -> Modal($G$, $w$, $n.\child_1$)
      $\MBoxx$  -> any (LAM$z$. Modal($G$, $\no(z)$, $n.\child_1$)) $\{y\in G \,:\, E(w,y)\}$
      $\MDiamm$  -> all (LAM$z$. Modal($G$, $\no(z)$, $n.\child_1$)) $\{y\in G \,:\, E(w,y)\}$
      $x$  -> $x\notin P(w)$
\end{lstlisting}
  \end{minipage}
\caption{Dual clause for $\modal$, which evaluates formula $\varphi$
  pointed to by $n$ against $G$ and verifies $G\not\models\varphi$.}
  \label{fig:modal-clause2}
\end{figure}

In light of the meta-theorem and the $\modal$ program, we have the
following.
\begin{theorem}
  Modal logic separation for sets of Kripke structures $\Pos$ and
  $\Neg$ with grammar $\grammar$ is decidable in time
  $\Oo(2^{\poly(mn)}\cdot|\grammar|)$, where
  $n = \max_{G\in \Pos\cup \Neg} |G|$ and $m = |\Pos| + |\Neg|$.
\end{theorem}
\begin{proof}[Proof Sketch.]
  For all Kripke structures $G=(W,s,E,P)$, $w\in W$, and formulas
  $\varphi$, we have that $G,w\models\varphi$ iff
  $(G, \route(\varphi), \clauses(\modal), w) \Downarrow_p$ and
  $G,w\not\models\varphi$ iff
  $(G, \route(\varphi), \clauses(\modal), \no(w)) \Downarrow_p$. The
  proof is by induction on $\varphi$. The rest follows by
  \Cref{lemma:mainlemma} and \Cref{corollary:complexity-lemma}, with
  $D=\{\tru,\fals\}$, $\sigma_\tru = s$ and $\sigma_\fals=\no(s)$,
  noting that $|\Asp(G)|=\Oo(|W|)$.
\end{proof}


\emph{Computation tree logic.} Do other modal logics over finite
Kripke structures have decidable learning? For computation tree logic
(CTL) the answer is affirmative, and we can program a $\lang$
evaluator whose aspects again involve the nodes of the Kripke
structure. We consider the following grammar for CTL formulas, from
which other standard operators can be defined:
\begin{align*}
  \varphi\Coloneqq a\in\Sigma \grammarsep \varphi\vee\varphi'\grammarsep
  \neg\varphi\grammarsep \EG\varphi\grammarsep \Ex(\varphi\Until\varphi')\grammarsep \EX\varphi
\end{align*}
These formulas are interpreted over finite Kripke structures, and the
operations in common with propositional modal logic are interpreted in
the same way. The novelty of CTL is that it can quantify over
\emph{paths} in the Kripke structure using the formulas starting with
$\Ex$, which assert the existence of a path along which the subformula
holds. The semantics for path quantifiers is given recursively based
on the following:
\begin{align*}
  G,w\models \EX\varphi \,\,\Leftrightarrow\,\, \exists w'.\, E(w,w')\text{ and }
  G,w'\models\varphi
\end{align*}
with the other two path quantifiers interpreted according to the
equivalences
\begin{align*}
  (1)\quad \EG\varphi \equiv \varphi \wedge \EX(\EG\varphi)
  \quad\text{and}\quad (2)\quad
  \Ex(\varphi\Until\varphi') \equiv \varphi' \vee \varphi \wedge \EX (\Ex(\varphi\Until\varphi')),
\end{align*}
with $(1)$ understood as a greatest fixpoint and $(2)$ as a least
fixpoint. This recursion introduces a subtlety, because the evaluator
should avoid infinite recursion caused by interpreting $\Ex$ formulas
as the right-hand sides of these equivalences. We can address this by
introducing a \emph{bounded counter} in the state of our CTL program,
which enables it to terminate after a sufficient amount of recursion
(bounded by the number of nodes in the Kripke structure). Below we
state decidable learning for CTL; see \betterappref{sec:ctl}~ for the program and
proof.

\begin{theorem}
  CTL separation for finite sets $\Pos$ and $\Neg$ of finite pointed
  Kripke structures, and grammar $\grammar$, is decidable in time
  $\Oo(2^{\poly(mn^2)}\cdot|\grammar|)$, where $m=|\Pos|+|\Neg|$ and
  $n=\max_{G\in \Pos\cup\Neg}|G|$.
\end{theorem}

\smallskip In the remainder of the paper we derive new decision
procedures for several learning problems by writing $\lang$
programs. We avoid details about the $\lang$ language and focus
instead on the logic of the programs as well as the aspects needed to
accurately evaluate expressions.

\section{Learning Regular Expressions}
\label{sec:regular-expressions}

In this section we develop a decision procedure for learning regular
expressions from finite words. In contrast to propositional modal
logic, the semantics of regular expressions involves recursion in the
structure of expressions as well as recursion over the structures
(finite words) themselves.

\subsection{Separating Words with Regular Expressions}
\label{sec:separ-words-with}
Consider the following problem.

\begin{problem}[Regular Expression Separation]
  Given finite sets $\Pos$ and $\Neg$ of finite words over an alphabet
  $\Sigma$, and a grammar $\grammar$, synthesize a regular expression
  over $\Sigma$ that matches all words in $\Pos$, does not match any
  word in $N$, and conforms to $\grammar$, or declare none exist.
\end{problem}

We consider extended regular expressions from the following grammar.
\begin{align*}
  e \Coloneqq a\in\Sigma \grammarsep e\cdot e'
  \grammarsep e + e' \grammarsep e
  \cap e' \grammarsep e^* \grammarsep \neg e
\end{align*}
Recall that, to use the meta-theorem (\Cref{lemma:mainlemma}), we must
program an evaluator for regular expressions over fixed
$\Sigma$-words. The notion of evaluation here is \emph{membership} of
a word $w$ in the language of a regular expression $e$, i.e.
$w\models e \Leftrightarrow w\in L(e)$. This semantics has a
straightforward recursive definition, and it can be presented in terms
of an auxiliary relation for membership of \emph{subwords} of $w$:
\begin{align*}
  w\in L(e) \quad \Leftrightarrow \quad w,(1,|w|+1)\models e.
\end{align*}
In the notation above, $(l,r)$ indicates the subword $w(l,r)$ from
positions $l$ to $r-1$ inclusive, taking $w(i,i)=\epsilon$ for any $w$
and $i$. The semantics of subword membership is given below.

\vspace{0.1in}
\begin{minipage}[t]{\textwidth}
\centering
\begin{tabularx}{\textwidth}{l l l l l l l}
  $w, (l,r)$ & $\models$ & $a \in \Sigma$ & if & $w(l) = a$ & and & $r
  = l+1$ \\
  $w, (l,r)$ & $\models$ & $e \cdot e'$ & if & $w, (l,k)\models e$ &
  and &
  $w, (k,r)\models e'$ \ for some $k\in [l,r]$\\
  $w, (l,r)$ & $\models$ & $e + e'$ & if & $w, (l,r)\models e$ & or
  & $w, (l,r)\models e'$ \\
  $w, (l,r)$ & $\models$ & $e\cap e'$ & if & $w, (l,r)\models e$
  & and & $w, (l,r)\models e'$ \\

  $w, (l,r)$ & $\models$ & $e^*$ & if & $l=r$ & or & $\exists
  k\in[l+1,r].\,\, w, (l,k)\models
  e$ \ and \ $w, (k,r)\models e^*$ \\

  $w, (l,r)$ & $\models$ & $\neg e$ & if & $w, (l,r)\not\models
  e$ && \\
\end{tabularx}
\end{minipage}
\vspace{0.1in}

\noindent Observe that if we \emph{fix} the word $w$ then the number
of pairs $(l,r)$ used in the definition above is finite. Also observe
that the definition in the case for Kleene star is well-founded
because either the expression size decreases or the subword length
decreases.

\paragraph{Remark} Regular expression separation has two
overfitting-style solutions if we ignore the syntax restriction from
the input grammar $\grammar$. Simply use $+_{w\in\Pos}w$ for the
tightest regular expression that matches all of $\Pos$, or
alternatively, $\cap_{w\in\Neg}\neg w$ for the loosest one that avoids
matching any of $\Neg$. If there is any separating regular expression
at all, then either of these must
work.

\subsection{Decidable Learning for Regular Expressions}
\label{sec:progr-eval-regul}

We write a $\lang$ program called $\reg$ which reads a regular
expression syntax tree and verifies whether a given word is a member
of the language for the regular expression. The class $\Mm$ consists
of structures encoding $\Sigma$-words, $\Ll$ consists of regular
expressions over $\Sigma$, and semantics is membership in the language
of regular expressions. States are pairs of ordered indices,
representing subwords, along with duals to handle negation. 
\begin{align*}
  \Asp(w) \coloneqq \left\{(l,r),\no(l,r) \,:\, l\le r \in [1, |w|+1]\right\}.
\end{align*}
\noindent For instance, if $w = \mathit{abbb}$, then the subword
$w'=\mathit{ab}$ is represented as the pair of positions $(1,3)$. The
alphabet $\Delta$ for syntax trees is straightforward and uses symbols
of arity $0$ for members of $\Sigma$.  We use state functions for
looking up the letter at a given position $i$, written $w(i)$, the
successor function on positions $x$, written $x+1$, and functions
$[x,y]$ and $[x+1,y]$ for computing the indices between two positions
$x\le y$, with $[x+1,y] = \emptyset$ if $x=y$. Boolean functions
include equality and disequality on positions and letters of $\Sigma$.

The program $\reg$ (with dual omitted) is given in
\Cref{fig:reg-clause1}. States matching $(l,r)$ are used by the
program to check whether $w(l,r)\in L(e)$, and states matching
$\no(l,r)$ are used to check whether $w(l,r)\notin L(e)$.  Using
$\reg$ we get the following.

\begin{figure}
  \centering
\begin{lstlisting}
  Reg($w$, $(l,r)$, $n$) = match $n.\lab$ with
     $*$  -> if ($l = r$) then True else
              any (LAM$x$. Reg($w$, $(l,x)$, $n.\child_1$) and Reg($w$, $(x,r)$, $n.\stay$)) $[l+1,r]$
     $\,\cdot$  -> any (LAM$x$. Reg($w$, $(l,x)$, $n.\child_1$) and Reg($w$, $(x,r)$, $n.\child_2$)) $[l,r]$
     $+$  -> Reg($w$, $(l,r)$, $n.\child_1$) or Reg($w$, $(l,r)$, $n.\child_2$)
     $\neg$  -> Reg($w$, $\no(l,r)$, $n.\child_1$)
     $\cap$  -> Reg($w$, $(l,r)$, $n.\child_1$) and Reg($w$, $(l,r)$, $n.\child_2$)
     $x$  -> $r = l+1$ and $w(l) = x$
\end{lstlisting}
\caption{$\reg$ evaluates the regular expression $e$ pointed to by $n$
  against an input word $w$ and verifies $w\in L(e)$.}
  \label{fig:reg-clause1}
\end{figure}

\begin{theorem}
  Regular expression separation for sets of words $\Pos$ and $\Neg$
  and grammar $\grammar$ is decidable in time
  $\Oo(2^{\poly(mn^2)}\cdot |\grammar|)$, where
  $n = \max_{w\in \Pos\cup \Neg} |w|$ and $m = |\Pos| + |\Neg|$.
\end{theorem}
\begin{proof}[Proof Sketch.]
  For all words $w$, positions $1\le i\le j\le |w|+1$, and regular
  expressions $e$, we have that $w(i,j)\in L(e)$ if and only if
  $(w, \route(e), \clauses(\reg), (i,j)) \Downarrow_p$ and
  $w(i,j)\notin L(e)$ if and only if
  $(w, \route(e), \clauses(\reg), \no(i,j)) \Downarrow_p$. The proof
  is by induction on $|w|$ and inner induction on $e$. We have
  $|\Asp(w)| = \Oo(|w|^2)$, and the theorem follows by
  \Cref{lemma:mainlemma} and \Cref{corollary:complexity-lemma}.
\end{proof}

\section{Linear Temporal Logic}
\label{sec:ltl}

In this section we consider synthesizing linear temporal logic (LTL)
formulas that separate infinite, periodic words. We again derive a
decision procedure for learning by writing a program.

\subsection{Separating Infinite Words with Linear Temporal Logic}
\label{sec:separ-infin-words}

We consider separating infinite periodic words over a finite alphabet
$\Sigma$. Such words $w\in\Sigma^\omega$ can be represented finitely
as the concatenation of a finite prefix $u\in\Sigma^*$ with a finite
\emph{repeated} suffix $v\in\Sigma^*$.  
For example, the infinite word
  $\mathit{babaabbaabbaabbaabb}\cdots$
  can be represented with $u=\mathit{bab}$ and $v=\mathit{aabb}$. The
  word is determined by the pair $(u,v)$, though in general there may
  be multiple ways to pick $u$ and $v$. Here we could also have picked
  $u=\mathit{baba}$ and $v=\mathit{abba}$. We refer to infinite
  periodic words, represented by pairs $(u,v)$, as \emph{lassos}.

\begin{problem}[Linear Temporal Logic Separation]
  Given finite sets $P$ and $N$ of lassos over an alphabet $\Sigma$,
  and a grammar $\grammar$ for LTL over $\Sigma$, synthesize a formula
  $\varphi\in L(\grammar)$ such that $w\models\varphi$ for all
  $w\in P$ and $w\not\models\varphi$ for all $w\in N$, or declare no
  such formula exists.
\end{problem}

Our LTL formulas come from the following grammar.
\begin{align*}
  \varphi \Coloneqq a\in\Sigma \grammarsep \varphi \wedge
  \varphi' \grammarsep \varphi \vee \varphi' \grammarsep \neg\varphi \grammarsep \Next \varphi
  \grammarsep \varphi \Until \varphi'
\end{align*}

Below we present the semantics in a way that makes clear the aspects,
which are \emph{positions} in the lasso $(u,v)$, along with an
indication of whether the position corresponds to $u$ or $v$. An LTL
formula $\varphi$ is true in a lasso $(u,v)$, written
$(u,v)\models\varphi$, precisely when $(u,v),(1,\_)\models\varphi$,
with the latter defined below. The main point is that the following two
relationships hold
\begin{align*}
  (u,v),(i,\_)\models\varphi \,\Leftrightarrow\,
                                  u^iv^\omega\models \varphi \quad\text{and}\quad
  (u,v),(\_,j)\models\varphi \,\Leftrightarrow\,
                                   v^{j}v^\omega\models \varphi,
\end{align*}
where $uv^\omega\models\varphi$ is the standard semantics for
LTL~\cite{temporal-logic-pnueli}. To denote the
letter at position $i$ we write $w(i)$. We use $i$ to range over
$[1,|u|]$ and $j$ to range over $[1,|v|]$
. If $j' < j$, then $[j,j']$ means $[1,j']\cup [j,|v|]$. We use
$[a,b)$ to exclude $b$.

\vspace{0.1in}
\begin{minipage}[t]{\textwidth}
\hspace{-0.15in}
\begin{tabularx}{\textwidth}{l l l l l l l}
  $(u,v),(i,\_)$ & $\models$ & $a \in \Sigma$ & if & $u(i) = a$ &&\\
  $(u,v),(\_,j)$ & $\models$ & $a \in \Sigma$ & if & $v(j) = a$ &&\\
  $(u,v),p$ & $\models$ & $\neg\varphi$ & if & $(u,v),p\not\models\varphi$ &&\\
  $(u,v),p$ & $\models$ & $\varphi\wedge\varphi'$ & if &
  $(u,v),p\models\varphi$ \ and \ $(u,v),p\models\varphi'$ && \\
  $(u,v),p$ & $\models$ & $\varphi\vee\varphi'$ & if &
  $(u,v),p\models\varphi$ \ or \ $(u,v),p\models\varphi'$ &&\\

  $(u,v),(|u|,\_)$ & $\models$ & $\Next\,\varphi$ & if &
  $(u,v),(\_,1)\models\varphi$ & &\\
  $(u,v),(i,\_)$ & $\models$ & $\Next\,\varphi$ & if &
  $(u,v),(i+1,\_)\models\varphi\quad i<|u|$ &\\
  $(u,v),(\_,j)$ & $\models$ & $\Next\,\varphi$ & if &
  $(u,v),(\_,\,j\,\,\text{mod}\,\, |v| + 1)\models\varphi$ &&\\\\

  $(u,v),(i,\_)$ & $\models$ & $\varphi\Until\varphi'$ & if & $\exists
  i'\ge i.\,\,\, (u,v),(i',\_)\models\varphi' \,\,\,\text{and}\,\,\, \forall
  i''\in[i,i').\,\,\, (u,v),(i'',\_)\models\varphi$ \\
  &&& or & $\exists j.\,\, \forall i'\ge i.\,\,\, (u,v),(i',\_)\models
  \varphi\,\,\,\text{and}\,\,\,\forall j'<j.\,\,\, (u,v),(\_,j')\models\varphi$\\
  &&&& $\,\,\qquad\quad\text{and}\,\,(u,v),(\_,j)\models\varphi'$ \\\\

  $(u,v),(\_,j)$ & $\models$ & $\varphi\Until\varphi'$ & if & $\exists
  j'.\,\,\,\forall j''\in [j,j').\,\,\,
  (u,v),(\_,j'')\models\varphi\,\,\,\text{and}\,\,\, (u,v),(\_,j')\models\varphi'$

\end{tabularx}
\end{minipage}
\vspace{0.1in}

\subsection{Decidable Learning for LTL}
\label{sec:progr-eval-ltl}
The $\lang$ program $\ltl$ in \Cref{fig:ltl-clauses} reads LTL syntax
trees and evaluates them over $\Sigma$-lassos, presented as pairs of
finite words $(u,v)\in\Sigma^*\times\Sigma^*$. Again we omit $\no$
clauses. The ranked alphabet $\Delta$ for syntax trees is similar to
those of regular expressions and modal logic. Signature functions
include functions for word length, written $|w|$, and functions for
computing sets of consecutive positions, e.g., $[x,y]$ and
$[x,y)$. There is also a function $\wraparound(j)$ defined by
\begin{align*}
  \wraparound(j) = \text{if}\,\, j > |v| \text{  then  } 1 \text{
  else  } j
\end{align*}
which is used to reset the current lasso position to the beginning of
the suffix $v$ when it exceeds $|v|$. There are also functions for
comparison of positions, e.g. $i < i'$, and equality and disequality
for alphabet letters at a given position, e.g. $u(i) = x$. The aspects
have the form $(\_,\cdot)$ and $(\cdot,\_)$ to encode whether a
position is part of $u$ or $v$. For a given lasso $(u,v)$, we have:
\begin{align*}
  \Asp((u,v)) &\coloneqq \left\{p, \no(p) \,:\, p\in\pos\right\}\\
  \pos &\coloneqq \left\{(i,\_) \,:\, i\in [1,|u|]\}\cup\{(\_,j) \,:\, j\in[1,|v|]\right\}
\end{align*}
\begin{theorem}
  Linear temporal logic separation for finite sets $\Pos$ and $\Neg$
  of lassos, and grammar $\grammar$, is decidable in time
  $\Oo(2^{\poly(mn)}\cdot|\grammar|)$, with $m=|\Pos|+|\Neg|$ and
  $n=\max_{(u,v)\in \Pos\cup \Neg}(|uv|)$.
\end{theorem}
\begin{proof}[Proof Sketch.]
  For each lasso $(u,v)\in \Pos$, positions
  $i\in[1,|u|], j\in[1,|v|]$, and LTL formula $\varphi$, we have that
  $(u,v),(i,\_)\models\varphi$ if and only if
  $((u,v), \route(\varphi), \clauses(\ltl), (i,\_))\Downarrow_p$ and
  $(u,v),(\_,j)\models\varphi$ if and only if
  $((u,v), \route(\varphi), \clauses(\ltl),
  (\_,j))\Downarrow_p$. Similarly, for each lasso $(u,v)\in \Neg$ we
  have that $(u,v),(i,\_)\not\models\varphi$ if and only if
  $((u,v), \route(\varphi), \clauses(\ltl), \no((i,\_)))\Downarrow_p$
  and $(u,v),(\_,j)\not\models\varphi$ if and only if
  $((u,v), \route(\varphi), \clauses(\ltl), \no((\_,j)))\Downarrow_p$. The proof is by induction on $\varphi$. We have
  $|\Asp((u,v))| = \Oo(|uv|)$ and the rest follows by
  \Cref{lemma:mainlemma} and \Cref{corollary:complexity-lemma}.
\end{proof}

\begin{figure}
  \centering
\begin{lstlisting}
  LTL($(u,v)$, $(i,\_)$, $n$) = match $n.\lab$ with

     conj -> LTL($(u,v)$, $(i,\_)$, $n.\child_1$) and LTL($(u,v)$, $(i,\_)$, $n.\child_2$)
     disj -> LTL($(u,v)$, $(i,\_)$, $n.\child_1$) or LTL($(u,v)$, $(i,\_)$, $n.\child_2$)
     neg -> LTL($(u,v)$, $\no(i,\_)$, $n.\child_1$)
     $\Next$ -> if $i < |u|$ then LTL($(u,v)$, $(i+1,\_)$, $n.\child_1$) else LTL($(u,v)$, $(\_,1)$, $n.\child_1$)

     $\Until$ -> any (LAM$i'.$ LTL($(u,v)$, $(i',\_)$, $n.\child_2$) and
                 all (LAM$i''.$ LTL($(u,v)$, $(i'',\_)$, $n.\child_1$)) $[i,i')$) $[i,|u|]$
           or    all (LAM$i'.$ LTL($(u,v)$, $(i',\_)$, $n.\child_1$)) $[i,|u|]$ and
                 any (LAM$j.$ all (LAM$j'.$ LTL($(u,v)$, $(\_,j')$, $n.\child_1$)) $[1,j)$
                              and LTL($(u,v)$, $(\_,j)$, $n.\child_2$)) $[1,|v|]$
     $x$ -> $u(i) = x$

\end{lstlisting}
\begin{lstlisting}

  LTL($(u,v)$, $(\_,j)$, $n$) = match $n.\lab$ with

     conj -> LTL($(u,v)$, $(\_,j)$, $n.\child_1$) and LTL($(u,v)$, $(\_,j)$, $n.\child_2$)
     disj -> LTL($(u,v)$, $(\_,j)$, $n.\child_1$) or LTL($(u,v)$, $(\_,j)$, $n.\child_2$)
     neg -> LTL($(u,v)$, $\no(\_,j)$, $n.\child_1$)
     $\Next$ -> LTL($(u,v)$, $(\_,\wraparound(j+1))$, $n.\child_1$)

     $\Until$ -> any (LAM$j'.$ LTL($(u,v)$, $(\_,j')$, $n.\child_2$) and
                all (LAM$j''.$ LTL($(u,v)$, $(\_,j'')$, $n.\child_1$)) $[j,j')$) $[1,|v|]$

     $x$ -> $v(j) = x$
\end{lstlisting}
\caption{$\ltl$ evaluates the LTL formula $\varphi$ pointed to by $n$ over
  lasso $(u,v)$ and verifies that $(u,v)\models \varphi$. }
  \label{fig:ltl-clauses}
\end{figure}

\section{Context-Free Grammars}
\label{sec:cfg}
In this section, we consider another problem involving separation of
finite words. The goal is to synthesize a context-free grammar that
generates all positively-labeled words and no negatively-labeled
words. We derive a decision procedure as before by writing a $\lang$
program.

\subsection{Separating Words with Context-Free Grammars}
\label{sec:separ-words-with}

\begin{problem}[Context-Free Grammar Separation]
  Given finite sets $\Pos$ and $\Neg$ of finite words over an alphabet
  $\Sigma$, as well as a (meta-)grammar $\grammar$, synthesize a
  context-free grammar $G$ over nonterminals $\nt$, terminals
  $\Sigma$, and axiom $S\in \nt$, such that $G\in L(\grammar)$ and
  $P\subseteq L(G)$ and $N \cap L(G) = \emptyset$, or declare no such
  grammar exists.
\end{problem}

The semantics of context-free grammars (CFGs) is standard: a word is
generated by a grammar if we can build a parse tree for it using the
productions. We want to represent CFGs as syntax trees and then write
a program that reads such trees and evaluates whether a given word is
generated by the represented grammar. The syntax trees can organize
productions along, say, the right spine, with their right-hand sides
in left children as suggested below.  Note that the ranked alphabet
$\Delta$ uses a binary symbol $\lhs(A)$ and nullary symbol $\rhs(A)$
for each $A\in\nt$ to distinguish between occurrences of $A$ in the
left and right-hand sides of a production. 
We also use a binary symbol $\topp(S)$ to distinguish the root of the
syntax tree, as well as a nullary symbol $\eend$ to signal the end of
productions along the right spine. Terminals $a\in\Sigma$ are
represented as nullary symbols $\term(a)$. See below with a grammar on
the left and its syntax tree on the right.

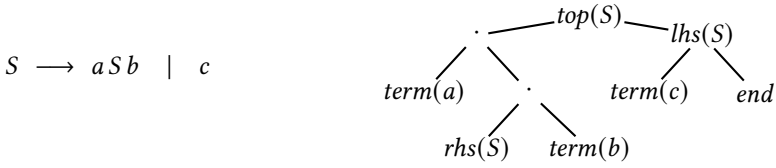
\begin{minipage}[t]{0.45\linewidth}
  \begin{align*}
    \\
    S \,\, &\longrightarrow\,\, a \, S \, b \quad|\quad c
  \end{align*}
\end{minipage}%
\begin{minipage}[t]{0.45\linewidth}
  \begin{align*}
    \begin{tikzpicture}[thick]
      \node[draw=none] (S1) at (0,0) {$\topp(S)$} ;
      \node[draw=none] (S2) at (1.5,-0.25) {$\lhs(S)$} ;
      \node[draw=none] (Dot) at (-1.5,-0.25) {$\cdot$} ;
      \node[draw=none] (A) at (-2.2,-1) {$\term(a)$} ;
      \node[draw=none] (S3) at (-0.8,-1) {$\cdot$} ;
      \node[draw=none] (rhsS) at (-1.5, -1.75) {$\rhs(S)$} ;
      \node[draw=none] (termb) at (0, -1.75) {$\term(b)$} ;
      \node[draw=none] (epsilon) at (0.8,-1) {$\term(c)$} ;
      \node[draw=none] (end) at (2.2,-1) {$\eend$} ;
      \draw [-] (S3) edge[black] (rhsS) ;
      \draw [-] (S3) edge[black] (termb) ;
      \draw [-] (S1) edge[black]  (S2) ;
      \draw [-] (S1) edge[black] (Dot) ;
      \draw [-] (Dot) edge[black] (A) ;
      \draw [-] (Dot) edge[black] (S3) ;
      \draw [-] (S2) edge[black] (epsilon) ;
      \draw [-] (S2) edge[black] (end) ;
    \end{tikzpicture}
  \end{align*}
\end{minipage}

\subsection{Decidable Learning for Context-Free Grammars}
\label{sec:progr-eval-cont}
We want a program that evaluates a CFG syntax tree $G$ to verify
whether $w\in L(G)$ for an input word $w$. What kind of state is
needed? Intuition suggests the $\lang$ evaluator will be similar to
the one for regular expressions, and that we should use pairs of
positions. The main difference is the more flexible recursion afforded
by nonterminals. Consider reading the syntax tree above starting at
the root labeled by $\topp(S)$, with \dblqt{$a\,S\, b$} in the left
subtree and the rest of the productions on the right. To verify that
$w$ is generated by $S$, the program should move to the right-hand
sides of the two $S$-productions and check whether $w$ is generated by
\emph{either} of these. Concatenation in the right-hand sides of
productions can be handled just like for regular expressions by
guessing a split for $w$ and then verifying the guess in the
subtrees. But upon reading, say, $\rhs(S)$ in the subtree for
\dblqt{$a\,S\,b$}, the program should navigate \emph{up} to find the
$S$ productions and enter the subtrees corresponding to their
right-hand sides in order to parse the current subword and verify its
membership in $L(S)$. This is accomplished by entering a state
$\reset(S)$ that causes the program to navigate to the root of the
syntax tree and then move downward in a state $\find(S)$ to find and
enter the right-hand sides of all productions for $S$.

The recursion afforded by the nonterminals introduces a subtlety when
verifying \emph{non-membership} of a word in the grammar (similar to
the subtlety discussed for CTL path quantifiers in
\Cref{sec:modal-expr-separ}). We return to this point after
formalizing the membership checking part of the program.

\subsubsection{Verifying Membership}
\label{sec:verify-memb}
We write a program $\cfg$ that evaluates an input CFG syntax tree $G$
over a word $w$ and verifies that $w\in L(G)$. The states consist of
ordered pairs of word positions as well as some extra information
related to moving up and down on the syntax tree, along with duals:
\begin{align*}
  \Asp(w)  &\coloneqq \{x, \dual(x) \,:\, x\in X(w)\}\\
  X(w) &\coloneqq \subs(w)\,\cup\,\{ (s,\find(A)), (s,\reset(A)) \,:\, s\in\subs(w),\, A\in\nt \}\\
  \subs(w) &\coloneqq \{ (l,r) \,:\, 1\le l\le r \le |w|+1 \}
\end{align*}
\noindent
Signature functions are the same as those for regular expressions. The
clauses for $\cfg$ (duals omitted) are shown in \Cref{fig:cfg}. The
program is designed to start from the state $\reset(S)$, from which it
proceeds to find and enter each of the productions for the starting
nonterminal $S$.

\begin{figure}
  \centering
  \begin{minipage}[c]{0.95\linewidth}
\begin{lstlisting}
  CFG($w$, $(l,r)$, $n$) = match $n.\lab$ with
     $\cdot$ -> any (LAM$x$. CFG($w$, $(l,x)$, $n.\child_1$) and CFG($w$, $(x,r)$, $n.\child_2$) $[l,r]$
     $\rhs(z)$ -> CFG($w$, $(l,r)$, $\reset(z)$, $n.\up$)
     $\term(x)$ -> $r = l+1$ and $w(l) = x$

  CFG($w$, $(l,r)$, $\reset(z)$, $n$) = match $n.\lab$ with
     $\topp(z)$ -> CFG($w$, $(l,r)$, $n.\child_1$) or CFG($w$, $(l,r)$, $\find(z)$, $n.\child_2$)
     $\topp(x)$ -> CFG($w$, $(l,r)$, $\find(z)$, $n.\child_2$)
     $\_$     -> CFG($w$, $(l,r)$, $\reset(z)$, $n.\up$)

  CFG($w$, $(l,r)$, $\find(z)$, $n$) = match $n.\lab$ with
     $\lhs(z)$ -> CFG($w$, $(l,r)$, $n.\child_1$) or CFG($w$, $(l,r)$, $\find(z)$, $n.\child_2$)
     $\lhs(x)$ -> CFG($w$, $(l,r)$, $\find(z)$, $n.\child_2$)
\end{lstlisting}
  \end{minipage}
\caption{$\cfg$ evaluates an input CFG syntax tree $G$ pointed to by
  $n$ against word $w$ and verifies that $w\in L(G)$.}
  \label{fig:cfg}
\end{figure}

Consider the operation of $\cfg$ over a word $w$ and a grammar $G$
that has a production like $A \rightarrow AA$. Notice that the program
could read this production arbitrarily many times in the same state by
always choosing to split $w$ into $\epsilon$ and $w$ when reading the
right-hand side $\dblqt{AA}$. This would cause it to verify
recursively that $\epsilon\in L(A)$ and $w\in L(A)$, which could
repeat again and again. Nevertheless, if indeed $w\in L(G)$, then
there is a finite proof for
$(w, \route(G), \clauses(\cfg), ((1, |w|+1),\reset(S))) \Downarrow_p$
that can be obtained by following any correct derivation of $w$ from
the grammar. The case for $w\notin L(G)$ is more subtle.

\subsubsection{Verifying Non-Membership}
\label{sec:non-membership}
Suppose now that $w\notin L(A)$ for a production like
$A \rightarrow AA$. How should $\cfg$ verify this? There is no
derivation to follow, and the program might loop forever by entering
the right-hand side and reading $AA$, which will cause it to read all
$A$ productions, which will cause it to read $AA$, and so on, with no
guarantee that subwords become smaller in each recursive call. This
termination issue can be dealt with in a few ways, e.g., by adding
states to keep track of the depth of recursion. But there is a
simpler solution.

It turns out that the duals for the $\cfg$ clauses in \Cref{fig:cfg}
are sufficient for verifying non-membership, provided all input
grammars are in \emph{Greibach normal form} (GNF). Productions in GNF
grammars have the form $A\rightarrow a(\nt)^*$, with $a\in \Sigma$ and
$A\in \nt$. Intuitively, this restriction helps because it makes
proofs of non-membership finite: subwords must become smaller each
time the program recursively checks a given nonterminal. To see this,
consider verifying $\mathit{aba}\notin L(S)$ for the GNF grammar:
\begin{align*}
  S \,\,\longrightarrow \,\, a\,S \,\,\,\,|\,\,\,\, b
\end{align*}
The word $\mathit{aba}$ is clearly not generated by the second
production. To show it is not generated by the first, we show there is
no way to split $\mathit{aba}$ into $w_1w_2$ so that $w_1$ is
generated by $a$ and $w_2$ is generated by $S$. If $w_1\neq a$ then
the subproof for that split can end. Otherwise $w_1 = a$, and thus
$|w_2| < |w|$, and hence the subproof for $w_2\notin L(S)$ will be
finite by induction on word length. For any GNF grammar $G$ and word
$w\notin L(G)$, there is a finite proof for
$(w, \route(G), \clauses(\cfg), \no((1, |w| + 1),\reset(S)))
\Downarrow_p$, and the argument does not in fact rely on the precise
form of GNF productions; it works for more general productions of the
form $A\rightarrow \alpha$ with
$\alpha\in (\Sigma\sqcup\nt)^*\Sigma(\Sigma\sqcup\nt)^*$, i.e., those
that involve at least one terminal. 
We call a grammar \emph{productive} if each of its productions meets
this requirement. As long as input grammars are productive, the dual
clauses for those from \Cref{fig:cfg} correctly verify
non-membership. Note that all context-free languages can be
represented by productive CFGs, provided we append an empty production
to include the empty word if needed. Therefore, we assume that the
input meta-grammar $\grammar$ encodes only productive
CFGs\footnote{
  Alternatively, we can use another automaton to verify that input
  trees encode productive grammars. The product of this automaton with
  the meta-grammar automaton $\aut_\grammar$ can itself be viewed as a
  meta-grammar which enforces productivity.
}.

\begin{theorem}
  CFG separation for finite sets of words $\Pos$, $\Neg$, and grammar
  $\grammar$ enforcing productivity, is decidable in time
  $\Oo(2^{\poly(k)}\cdot |\grammar|)$, where $k=mn^2\cdot|\nt|$,
  $n = \max_{w\in P\cup N}|w|$, and $m = |P| + |N|$.
\end{theorem}
\begin{proof}[Proof Sketch.]
  Fix a word $w$. For every $(i,j)\in \subs(w)$ and for every $G$ we
  have that $w(i,j)\in L(G)$ if and only if
  $(w,\route(G),\clauses(\cfg),
  ((i,j),\reset(S)))\Downarrow_p$. Similarly, we have
  $w(i,j)\notin L(G)$ if and only if
  $(w,\route(G),\clauses(\cfg),
  \no((i,j),\reset(S)))\Downarrow_p$. The proof is by induction on $w$
  and $G$. We have $|\Asp(w)| = \Oo(|w|^2\cdot |\nt|)$ and the theorem
  follows by \Cref{lemma:mainlemma} and \Cref{corollary:complexity-lemma}.
\end{proof}

\section{First-Order Logic over Rational Numbers with Order}
\label{sec:rationals}
In this section, we consider learning first-order logic queries over
an infinite domain, namely, the structure $(\Rat, <)$ consisting of
the rational numbers $\Rat$ with the usual linear order $<$. The
learning problem requires labeled $k$-tuples of rational numbers to be
separated by a \emph{query} in $\FOk$, i.e., a formula in first-order
logic with $k$ variables. We derive a decision procedure by writing a
$\lang$ evaluator for $\FOk$ over $(\Rat,<)$.

\subsection{Learning Queries over Rational Numbers with Order}
\label{sec:learn-quer-over}
We consider the following problem.
\begin{problem}[Learning $\FOk$ Queries over $(\Rat,<)$]
  Given finite sets $\Pos$ and $\Neg$ of $k$-tuples over $\Rat$ and a
  grammar $\grammar$ for $\FOk$ over $(\Rat,<)$, synthesize
  $\varphi(\overline{x})\in L(\grammar)$ such that
  $P \subseteq \{t\in \Rat^k \,\,|\,\, (\Rat,<), t \models
  \varphi(\overline{x})\}$\footnote{ We are abusing notation and
    treating $t$, a $k$-tuple of rationals, as an assignment to the
    $k$, ordered free variables of $\varphi$.} and
  $N \subseteq \{t\in \Rat^k \,\,|\,\, (\Rat,<), t \not\models
  \varphi(\overline{x})\}$, or declare no such formula exists.
\end{problem}

\noindent A ranked alphabet $\Delta$ for $\FOk$ has, for any variables
$x,y$, the unary symbols $\dblqt{\forall x}$, $\dblqt{\exists x}$ and
nullary symbols $\dblqt{x < y}$, $\dblqt{x=y}$, in addition to the
symbols for Boolean operations. 
Note that, in this problem, the class of structures $\Mm$ for the
language is the set $\Rat^k$, and so a single ``structure'' is a tuple
of rationals $t\in\Rat^k$. For $t\in \Rat^k$, the semantics is given
by $t\models \varphi \Leftrightarrow (\Rat,<), t\models\varphi$.

\subsection{Decidable Learning for First-Order Logic Queries over
  Rationals}
\label{sec:progr-eval-first}
Given $t\in \Rat^k$ and $\varphi\in\FOk$, what semantic information do
we need to verify $(\Rat,<), t\models\varphi$?  Consider evaluating a
query $\varphi(x,y,z)$ on
$t=(\nicefrac{1}{2},3, \nicefrac{4}{3})\in P$. Our program might start
with an assignment $\gamma$ that lets it remember $(x,y,z)$ maps to
$(\nicefrac{1}{2},3, \nicefrac{4}{3})$. With the right functions, it
can easily verify atomic formulas by simply checking
$\gamma(x) < \gamma(y)$ or $\gamma(x) = \gamma(y)$. When the program
reads, say, \dblqt{$\exists x$}, it must carry forward some finite
amount of information, which thus excludes tracking the precise values
for the variables, of which there are infinitely many.

The main idea is that evaluating atomic formulas does not require the
precise values of the variables: the \emph{order between variables} is
all that is needed to evaluate $\FOk$ formulas over $(\Rat,
<)$. 
  In our example, we have
  $t=(\nicefrac{1}{2},3, \nicefrac{4}{3})$ corresponding to the
  assignment
  $\{x\mapsto \nicefrac{1}{2}, y\mapsto 3, z\mapsto
  \nicefrac{4}{3}\}$, and so the program begins in a state encoding
  that $x < z < y$. Suppose it reads the formula
  $\exists x\forall y\, (x < y)$. First it reads \dblqt{$\exists x$}
  and branches (disjunctively) on all of the finitely-many
  \emph{distinguishable} choices for where to place $x$ relative to
  the other variables while leaving the others in the same relative
  positions. We preserve $z < y$, but $x$ can appear in any of several
  positions: $x < z < y$, $x=z < y$, $z < x < y$, $z < x=y$,
  $z < y < x$. From each of these states the program reads
  \dblqt{$\forall y$} and branches (conjunctively) on every choice for
  where to place $y$. It eventually rejects the formula because $y$
  can always be placed strictly below $x$ in the second branching
  step.

  \Cref{fig:rat} shows a program $\rat$ that evaluates $\FOk$ formulas
  over tuples of rational numbers.  The states of $\rat$ record an
  ordering between $k$ variables from a set $V$, including whether two
  variables are equal, and thus they correspond to the \emph{total
    preorders} on $V$, denoted $\pre(V)$. We use \dblqt{$\gtrsim$} as
  a pattern variable to denote a preorder. For a given tuple $t$ we
  have
\begin{align*}
  \Asp(t) \coloneqq \left\{\, \gtrsim,\, \no(\gtrsim) \,:\,\,\,
  \gtrsim\,\,\in\pre(V)\,\right\}.
\end{align*}
Signature functions include Boolean functions for checking ordering
and equality in a given preorder $\gtrsim$, which we denote by
$\islt(x,z,\gtrsim)$ and $\iseq(x,z,\gtrsim)$ and define by:

\vspace{0.1in}
\begin{minipage}[t]{\textwidth}
\centering
\begin{tabularx}{0.85\textwidth}{l l l l l l l}
  &$\isgeq(x,z,\gtrsim)$ &$\coloneqq$ &$x \gtrsim z$ \qquad\qquad
  &$\iseq(x,z,\gtrsim)$ &$\coloneqq$ &$x \gtrsim z \,\,\wedge\,\, z
                       \gtrsim x$ \\
  &$\islt(x,z,\gtrsim)$ &$\coloneqq$ & 
 $\neg\isgeq(x,z,\gtrsim)$\qquad\qquad
  &$\isneq(x,z,\gtrsim)$ &$\coloneqq$ &$\neg\iseq(x,z,\gtrsim)$
\end{tabularx}
\end{minipage}\hfill
\vspace{0.05in}

\noindent State functions include $\place(x,\gtrsim)$, which computes
the set of all total preorders that place $x\in V$ in a new position
but agree with $\gtrsim$ on variables in $V\setminus \{x\}$, defined
as:
\begin{align*}
  \place(x,\gtrsim) \coloneqq \,\, \left\{\, \gtrsim'\,\in\,\pre(V)\,\,:\,\,
  y \gtrsim' z \Leftrightarrow y \gtrsim z,\, \forall y,z\in
  V\setminus\{x\} \,\right\}
\end{align*}
\begin{theorem}
  Learning queries in $\FOk$ over the rational numbers with order,
  with sets of $k$-tuples $\Pos$ and $\Neg$ and grammar $\grammar$, is
  decidable in time $\Oo(2^{\poly(mk^k)}\cdot|\grammar|)$, where
  $m=|\Pos|+|\Neg|$.
\end{theorem}
\begin{proof}[Proof Sketch]
  Follows reasoning from previous sections and uses
  \Cref{lemma:mainlemma} and \Cref{corollary:complexity-lemma}. For $t\in\Rat^k$
  and a set $V$ of $k$ variables, we have
  $|\Asp(t)| = 2|\pre(V)| = \Oo(k^k)$.
\end{proof}

\begin{figure}
  \centering
  \begin{minipage}[c]{0.7\linewidth}
\begin{lstlisting}
  Rat($t$, $\gtrsim$, $n$) = match $n.\lab$ with
     $\forall x$ -> all (LAM$p$. Rat($t$, $p$, $n.\child_1$)) $\place(x,\gtrsim)$
     $\exists x$ -> any (LAM$p$. Rat($t$, $p$, $n.\child_1$)) $\place(x,\gtrsim)$
     conj -> Rat($t$, $\gtrsim$, $n.\child_1$) and Rat($t$, $\gtrsim$, $n.\child_2)$
     disj -> Rat($t$, $\gtrsim$, $n.\child_1$) or Rat($t$, $\gtrsim$, $n.\child_2)$
     neg -> Rat($t$, $\no(\gtrsim)$, $n.\child_1$)
     $x < z$ -> if $\islt(x, z, \gtrsim)$ then True else False
     $x = z$ -> if $\iseq(x, z, \gtrsim)$ then True else False
\end{lstlisting}
\vspace{-0.1in}
  \end{minipage}
\caption{$\rat$ evaluates formula $\varphi$ pointed to by $n$ against
  a tuple $t$ of rational numbers and verifies
  $(\Rat,<), t \models\varphi$.}
  \label{fig:rat}
\end{figure}

\paragraph{Remark}
Structures like $(\Rat, <)$ have a special kind of automorphism group,
called an \emph{oligomorphic group}
(see~\cite{hodges93}). Oligomorphic automorphism groups have
finitely-many orbits in their action on $k$-tuples from the domain of
the structure, for every $k$. In the case of learning queries over
$(\Rat, <)$, the number $k$ is fixed, and the $\lang$ program
evaluates formulas by keeping track of these finitely-many orbits. It
checks atomic formulas in a given orbit represented by a total
preorder on variables, and when evaluating quantifiers it is able to
move to all ``nearby'' orbits. There are many other examples of such
structures, e.g. those appearing in constraint satisfaction problems
over infinite domains (see~\cite{bodirsky-csp}). It would be
interesting to explore decidable learning results in more domains like
these.

\section{Decidable Learning for String Programs}
\label{sec:discussion}

In this section, we consider Gulwani's language for string
programming~\cite{flashfill}, which is designed to express
transformations of a sequence $i$ of input strings into an output
string $o$ in the context of spreadsheets. This language, which we
refer to as $\str$, turns out to be FAC, with the caveat that loops
must use variables from a finite set. We next give an overview of the
syntax and semantics of $\str$; details can be found in the original
paper~\cite{flashfill}. Then we discuss how to implement a $\lang$
evaluator that reads $\str$ syntax trees and checks whether they map
an input sequence $i$ to an output $o$. The language, being one used
in practice, is considerably more complex than our other examples, and
so we only sketch the main ideas.

\subsection{$\str$ Overview}
\label{sec:str-overview}
Programs in $\str$ map finitely-many input strings $v_j$ to an output
$o$. A \emph{program} $P\in\str$ consists of a switch statement
$\swit((\varphi_1, e_1),...\,,(\varphi_n, e_n))$ that chooses the
expression $e_i$ whose condition $\varphi_i$ is the first in the
sequence that is true. The $\varphi_i$ are DNF formulas over atoms
$\mat(v_j,r,k)$, which hold if at least $k$ matches for a regular
expression $r$ can be found in the input $v_j$. The $e_i$ in the
switch statement have the form $\concat(f_1,...\,,f_n)$. They
concatenate expressions $f_i$ that come in three flavors: (1)
$\substr(v_j,p_1,p_2)$ selects the substring in $v_j$ between
positions $p_1$ and $p_2$, (2) $\constr(s)$ denotes a string literal
$s$, and (3) $\loope(\lambda x.\, e)$ iteratively appends the result
of evaluating $e$ until that result is $\bot$, which is a special
value for failure. During loop iteration $i$, the variable $x$ is
bound to $i$ in $e$.

The positions $p_i$ in $\substr(v_j,p_1,p_2)$ are either constant
integers $\cpos(k)$ or they have the form $\poss(r_1,r_2,c)$, where
the $r_i$ are regular expressions and $c$ is a linear integer
expression built from constants and loop variables, e.g. $2x + 3$. The
expression $\poss(r_1,r_2,c)$ is evaluated with respect to $v_j$, and
returns a position $t$ such that just to the left of $t$ in $v_j$
there is a match for $r_1$ and starting at $t$ there is a match for
$r_2$. Furthermore, it returns the $c^{\mathit{th}}$ such position, or
$\bot$ if not enough such positions exist. Note that regular
expressions were restricted in $\str$ to use Kleene star and
disjunction only in a particular way, which we ignore. It is no
trouble to write a $\lang$ evaluator for a generalization of $\str$
that allows unrestricted (extended) regular expressions like those
from \Cref{sec:regular-expressions}.

Consider a program that extracts capital letters of an input string
(\cite{flashfill}, example $5$).

\begin{center}
\begin{tabular}{| l || c |}
  \hline
  Input $v_1$ & Output $o$ \\ \hline\hline
  \textit{Principles Of Programming Languages} & \textit{POPL} \\ \hline
\end{tabular}
\begin{align*}
  \text{Program:} &\quad\loope(\lambda x.\,
  \concat(\mathsf{SubStr2}(v_1,\, \mathsf{UpperTok},\, x))) \\
  \text{where}&\quad\mathsf{SubStr2}(v_j,\, r,\, c) \,\equiv\, \substr(v_j,\, \poss(\epsilon, r, c),\, \poss(r, \epsilon, c))
\end{align*}
\end{center}
This program uses $\mathsf{SubStr2}(v_j, r, c)$ to compute the
$c^{\mathit{th}}$ match of the regular expression $r$ in $v_j$. This
is used to extract the $x^{\mathit{th}}$ upper case letter in
iteration $x$ of the loop, which is then appended to previously
extracted letters. The loop exits when the body evaluates to $\bot$,
which happens when there are no more matches for $\mathsf{UpperTok}$.

\subsection{Decidable Learning for $\str$}
\label{sec:decid-learn-str}
We describe how a $\lang$ program should evaluate the constructs in
$\str$. Fix a set of input strings $i=v_1,\ldots, v_n$ and an output
string $o$. The evaluator reads $P\in\str$ and checks that it maps $i$
to $o$. We mention specific choices for representing syntax trees as
needed.

The $\swit$ statement can be modeled with a ternary symbol
$\swit(\varphi,e,\mathsf{rest})$, where $\mathsf{rest}$ represents the
rest of the cases with nested operators of the same kind. Upon reading
$\swit$, the program branches to verify \emph{either} the conditional
$\varphi_1$ holds and $e_1$ produces $o$ \emph{or} $\varphi_1$ does
not hold and the rest of the $\swit$ produces $o$.

The DNF formulae $\varphi_i$ can be easily evaluated with the Boolean
operators in $\lang$. An atom $\mat(v_i,r,k)$ can be represented with
a binary symbol $\mat_{v_i}(r,k)$, one for each $v_i$, with right
child a \emph{unary} representation of integer $k$, i.e. $s^k(0)$. To
check $\mat_{v_i}(r,s^k(0))$, the program evaluates the right child to
determine the value of $k$. Crucially, it can reject if $k$ exceeds
$|\subs(v_i)|$, which upper bounds the maximum number of matches for
any regular expression over $v_i$. Having determined $k$, the program
can branch over all $\binom{N}{k}$ combinations of subwords that could
witness the requisite $k$ matches, with $N=|\subs(v_i)|$. For each
subword, the program executes $\reg$ (\Cref{fig:reg-clause1}) as a
subroutine to check whether it matches the regular expression in the
left child.

It remains to interpret $e$ and verify it produces an output $o$. We
represent $\concat(f_1,...\,,f_n)$ in a nested way like $\swit$, and
binary $\concat(f,f')$ is evaluated as for regular expressions by
branching on all ways to split $o$ (or one of its subwords) into
consecutive subwords $w$ and $w'$, with $f$ and $f'$ then verified to
produce $w$ and $w'$.

The expressions $f$ are verified to yield a given word as
follows. Literals $\constr(s)$ are represented with nested
concatenation and thus follow the same idea as
$\concat(f,f')$. Substrings $\substr(v_j,p_1,p_2)$ are modeled with
binary symbols $\substr_{v_j}(p_1,p_2)$, one for each $v_j$. Having
determined the values of positions $p_i$, the program can simply use a
function for equality of subwords. Constant positions $\cpos(k)$ are
determined as before except the program rejects if $k$ exceeds
$|v_j|$. To evaluate $\poss(r_1,r_2,c)$, represented as a ternary
operator, the program guesses a position $t$ in $v_j$ and verifies
existence of matches for $r_1$ and $r_2$ to the left and right of
$t$. It further verifies there are $c-1$, \emph{but not} $c$ such
positions to the left of $t$. This is accomplished by branching on the
possible $\binom{t-1}{c-1}$ combinations of positions and checking for
the requisite matches, and then checking the opposite for each of the
$\binom{t-1}{c}$ combinations. Finally, integer expressions
$c= k_1x + k_2$ 
can be evaluated by hardcoding rules for bounded arithmetic, because
the maximum value that loop variables can take is bounded by $|o|$,
which we discuss next.

Provided the number of loop variables is finite, the program can
evaluate loops using a map $\gamma$ from variables $\{x_i\}$ to
integers. The integers are bounded because the loop body $e$ must
produce a string of non-zero length (otherwise the loop terminates),
and loop expressions are only ever verified to produce words of length
no more than $|o|$. Since each iteration must productively decompose a
word of length bounded by $|o|$, we can use $|o|$ as a bound on the
range of $\gamma$. Thus the $\gamma$ have finite domain and range and
require finitely-many states. Now, suppose the program encounters a
loop $\loope(\lambda x.\, e)$ with current variable map $\gamma$, and
suppose it must verify the loop produces a word $w$. It first sets
$\gamma(x)=1$. Then it guesses a decomposition of $w$ into $w_1w_2$
such that $e$ evaluated with $\gamma$ produces $w_1$ and
$\loope(\lambda x.\, e)$ evaluated with
$\update{\gamma}{x}{\gamma(x)+1}$ produces $w_2$.

We conclude by \Cref{lemma:mainlemma} that learning $\str$ programs
from examples is decidable, even for the generalization that allows
unrestricted regular expressions (which~\cite{flashfill} disallows).


\section{Related Work}
\label{sec:related-work}

\emph{Expression Learning and Program Synthesis.} Our approach is
inspired by recent results for learning in finite-variable
logics~\cite{popl22}. 
  Proofs in that work involve direct
  automata constructions, and the results can be obtained with our
  meta-theorem by writing suitable evaluators. Our work generalizes
  the tree automata approach to general symbolic languages by
  separating decidable learning theorems into two parts: (1)
  identifying the underlying semantic aspects of the language in
  question and (2) programming with this new datatype in order to
  evaluate arbitrary expressions. The finite-variable restriction for
  the logics considered in~\cite{popl22} leads to finitely-many
  aspects— a finite set of assignments to some $k$ variables. But, as
  our work shows, this restriction is not necessary for decidable
  learning; several languages we consider do not use variables and it
  is unclear what a corresponding variable restriction would
  mean. Usual translations of regular expressions to monadic
  second-order logic formulae, for instance, do not stay within a
  finite-variable fragment. Nevertheless, the recursive semantics of
  regular expressions involves subwords, and there are only
  finitely-many subwords of a given word, which makes regular
  expressions finite-aspect checkable. 

Practical algorithms for some of the learning problems we address have
been explored previously, e.g. learning
LTL~\cite{neider-ltl-learning}, regular
expressions~\cite{jagadish-regexp-learning-08,fernau-regexp-learning-09},
and context-free
grammars~\cite{sakakibara-cfg-learning,langley-cfg-learning,vanlehn87-cfg-learning},
but decidable learning results with syntactic restrictions have not
been established.

Other recent work studies the parameterized complexity of learning
queries in first-order logic ($\FO$)~\cite{van-bergerem-22}, algorithms for learning in
$\FO$ with counting~\cite{learning-in-fo-with-counting}, and learning
in description logics~\cite{lutz-ijcai2019}. Applications for $\FO$
learning have emerged, e.g., synthesizing
invariants~\cite{aiken-fo-sep,parno-21,distAI,koenig-taming-search-space,ice,madhu-qda}
and learning program properties~\cite{inferring-rep-invariants,
  preconditions, learning-contracts}.

Expression learning is connected to program synthesis, and in
particular, programming by example~\cite{flashmeta}, where practical
algorithms have been used to automate tedious programming tasks, e.g.
synthesizing string programs~\cite{flashfill,flashfillplus},
bit-manipulating programs from templates~\cite{sketch}, or functional
programs from examples and type
information~\cite{type-directed-synth-osera,synth-refinement-types-nadia16}. Synthesis
with grammar restrictions follows work in the
\textsc{SyGus}~\cite{sygusJournal} framework. 

\emph{Automata for Synthesis.} Connections between automata and
synthesis go back to Church's problem~\cite{church63-journal} on
synthesizing finite state machines that manipulate infinite streams of
bits to meet a given logical specification. This was solved first by
B{\"u}chi and Landweber~\cite{BuchiLandweber69} for specifications in
monadic second-order logic, and later also by
Rabin~\cite{Rabin72}. The idea was to translate the specification into
an automaton, and to view synthesis of a transducer as the problem of
synthesizing a finite-state winning strategy in a game played on the
transition graph of the automaton. The result was a potentially large
transition system, not a compact program. The use of tree automata
that work over syntax trees was advanced in~\cite{madhuCSL11} and has
been used for practical algorithms in several program synthesis
contexts~\cite{fta-data-completion-scripts, Wang2017,
  WangWangDilligOOPSLA18,ecta,miltner-angelic,handa-rinard}.

\emph{Decidability in Synthesis.} Many foundational decidability
results in logic and synthesis of finite-state systems rely on
reductions to automata
emptiness~\cite{BuchiLandweber69,Rabin72,automata-logics-games,kpvPneuli,PR89,KMTV00,PR90}.
Recent decidability results for synthesis of uninterpreted programs
involved a reduction to emptiness of two-way tree
automata~\cite{uninterpretedsynth}. Decision procedures for
\textsc{SyGuS} problems in linear and conditional linear integer
arithmetic~\cite{farzan22,reps20-unrealizability} used grammar flow
analysis~\cite{gfa-1991} and an abstraction based on semi-linear
sets.

\emph{Automata for Learning vs. Graph Algorithms.} There is a large
body of work, e.g. see~\cite{Habel1992,courcelle}, on checking
properties of graphs expressed in monadic second-order logic. These
results involve translating logical properties into automata that read
\emph{decompositions of graphs} and accept if the represented graph
has the property. Our work is very differently motivated: we are
interested in properties of syntax trees defined over arbitrary
\emph{fixed} structures (e.g. unrestricted graphs like cliques or
grids), and the properties are motivated by semantics of complex
symbolic languages. Our automata constructions for \emph{learning}
are, conceptually, dual to these constructions from logical
specifications.

\emph{Definability in Monadic Second-Order Logic.} Foundational
results from logic and automata theory connect definability in monadic
second-order logic and recognizability by finite machines. These
results span various classes of structures, including finite words and
trees~\cite{weak-so-and-finite-aut-buchi, elgot-61,
  trakhtenbrot-61,doner-1970,thatcher-wright}, infinite words and
trees~\cite{Buchi1990, Rabin69}, and graphs with bounded tree
width~\cite{courcelle90}. It follows by definition that for any FAC
language, the semantics over any fixed structure can be captured by a
sentence in monadic-second order logic over syntax trees.


\section{Conclusion}
\label{sec:conclusion}
We introduced a powerful recipe for proving that a symbolic language
has decidable learning. It involves writing a program, i.e. semantic
evaluator, that operates over mathematical structures and expression
syntax trees for a given symbolic language. \emph{Finite-aspect
  checkable} languages have the property that the semantics of any
expression $e$ can be expressed in terms of a finite amount of
semantic information (aspects) that depends on the structure over
which evaluation occurs but not on the size of $e$. 
  This
  addresses a central question in expression learning with
  \emph{version space algebra} (VSA) techniques, especially those
  realized as tree automata: for which symbolic languages are these
  learning algorithms possible? One prevailing answer in the
  literature is that language operators should have \emph{finite
    inverses} (e.g., see~\cite{flashmeta,flashfillplus}). When
  operators have finite inverses a top-down tree automaton can be
  effectively constructed. Our work suggests a weaker requirement,
  namely, that it be sufficient to \emph{evaluate} arbitrarily large
  expressions or programs on specific examples by traversing syntax
  trees up and down using memory that is bounded by a function solely
  of the \emph{size of the example}. For instance, in
  \Cref{sec:rationals} we considered learning queries over the
  rational numbers with order. The ``inverse'' of $x < y$ is an
  infinite set of ordered pairs of rational numbers. Nevertheless, to
  evaluate a formula for a specific example, all that is needed is a
  bounded number of bits to encode the current ordering of variables.

We have also presented a set of interesting FAC languages that have
nontrivial semantic definitions using finitely-many aspects, and new
decidable learning results for each. We believe that many more can be
readily found using our meta-theorem.

Tree automata underlie many practical algorithms for synthesis based
on compactly representing large spaces of programs and
expressions~\cite{miltner-angelic,fta-data-completion-scripts,Wang2017,WangWangDilligOOPSLA18,ecta,handa-rinard,flashfill}. The
main idea is to efficiently represent classes of expressions which are
equivalent with respect to some examples. This idea originates with
version space algebra~\cite{tom-mitchell-vsa,mitchell97}, which
essentially amounts to a restricted form of tree automata working over
trees of bounded depth~\cite{vsa-ta}. Bringing the full tree automata
toolkit to bear on learning and synthesis, e.g. two-way power and
alternation, recognizes tree automata as a kind of basic building
block for a version space algebra \emph{over tree automata}. The
learning constructions from our work use semantic evaluators (compact,
effective descriptions of tree automata) as a basic building block and
combine them in specific ways to address specific learning
problems. Given this uniform technique of using tree automata as a
programming language, it would be interesting to build compact
representations and incremental algorithms for their construction and
emptiness that yield generic learning algorithms which scale. Bounding
the depth of expressions may make some of the constructions from this
paper feasible.

Tree automata have also been used in many other contexts in computer
science. In fixed-parameter tractable algorithms (e.g. Courcelle's
theorem~\cite{parameterized-complexity-flum-grohe}) and finite model
theory, they have been used to obtain generic algorithms for
$\MSO$-definable properties that work over tree decompositions of
graphs, while in temporal logic verification~\cite{vardi-ltl} they
have been used as acceptors of correct behaviors of systems. Their use
for learning in symbolic languages is an emerging new application of
tree automata. It would be interesting to study the \emph{theory} of
FAC languages in terms of expressiveness, language-theoretic
properties, and alternative characterizations.

\newpage
\bibliography{references}


\app{\newpage\appendix
\section{Detailed description of $\lang$}
\label{sec:app-language}
Programs in $\lang$ are parameterized by a language
$(\Ll,\Mm,\models)$, where the semantic function
$(\_\,\models\,\_) \,:\, \Ll\times\Mm\rightarrow D$ specifies the domain $D$ in
which expressions are interpreted. A program takes as input a
\emph{pointer} into the syntax tree for an expression $e\in\Ll$ as
well as a structure $M\in\Mm$. A program $P$ navigates up and down on
$e$ using a set of pointers to move from children to parent and parent
to children in order to evaluate the semantics of $e$ over the
structure $M$ and verify that $M\models e = d$ for some $d\in D$.

\subsection{Parameters}
\label{sec:parameters}
To write a $\lang$ program we specify two things: (1) the language
over which the program is to operate and (2) the program's auxiliary
\emph{states}, which correspond to semantic aspects, i.e. the
auxiliary information used in the definition of the language
semantics\footnote{
    We sometimes use \emph{states} and
    \emph{aspects} interchangeably. But, occasionally we use
    \emph{aspects} to distinguish auxiliary semantic information, with
    which we need not associate any operational meaning, from the
    operational meaning associated with \emph{states} in the context
    of a program.
}. Part (1) involves specifying (a) the syntax trees
for $\Ll$ in terms of a ranked alphabet $\Delta$ and (b) the signature
for structures $\Mm$, including a set of functions used to access the
data for a given $M\in\Mm$. Additionally, in part (a) we specify an
algebraic data type (ADT) that endows the symbols of $\Delta$ with
extra structure in order to allow programs to pattern match over the
alphabet. As an example, suppose we model universal quantification in
first-order logic by adding to $\Delta$ a unary symbol
\dblqt{$\forall x$} for each variable $x$ from some finite set of
variables. We could then treat \dblqt{$\forall$} as a unary
constructor to allow a program to match over all alphabet symbols that
represent a universally-quantified variable. We abuse notation and
write $\Delta$ for both the ranked alphabet and ADT when the context
is clear.

Part (2) is accomplished by specifying an ADT for the program
states. This ADT will typically be infinite, but any fixed structure
will use only a finite subset of it, provided the language is FAC. We
denote the state ADT by $\Asp$ and the subset pertaining to a given
structure $M$ by $\Asp(M)$. In some cases, $\Asp$ has little or no
structure. For instance, in modal logic it consists of nullary
constructors for the nodes of a Kripke structure. In other cases,
e.g., regular expressions (\Cref{sec:regular-expressions}), we use a
\emph{pair} constructor over positions in finite words with
$\Asp = \{(i,j)\in\Nat\times\Nat \,:\, i \le j\}$. We will clarify
these choices in each setting as needed.

The ADTs are each associated with a set of \emph{patterns}, which are
built using the ADT constructors together with variables from a set
$\varit$. We use $x,z\in\varit$ as pattern variables; these should not
be confused with variables used in expressions $\Ll$. The sets of
state and alphabet patterns are denoted by $\Asp(\varit)$ and
$\delvar$, respectively. Note that we do not assume $\Asp$ has a
finite signature. All we will need is that $\Asp(M)$ is finite for
every $M$ and that state and alphabet pattern matching is
computable. For the latter, we assume a computable function $\match$
that computes a unifying substitution for two members of
$\Asp(\varit)$ or $\delvar$ whenever possible.

\subsection{Semantics}
\label{sec:semantics}
The semantics for $\lang$ is given in \Cref{fig:big-step}. It defines
two relations
\begin{align*}
  (M, n, \clauses, \sigma) \Downarrow_p \qquad \text{and} \qquad (M,
  n, \clauses, e) \Downarrow_e.
\end{align*}
The predicate $\Downarrow_p$ holds for program configurations
$(M,n,\clauses,\sigma)$, with $\clauses$ being the set of clauses for
the program, $M$ a structure, $n$ a pointer, and $\sigma\in\Asp(M)$ a
state. The predicate $\Downarrow_e$ holds for expression
configurations $(M, n, \clauses, e)$, with $e$ being an expression
from the grammar in \Cref{fig:lang}. The assertion
$(M, n, \clauses, \sigma) \Downarrow_p$ can be read as follows: over
the structure $M$ and syntax tree pointed to by $n$, the program
consisting of clauses $\clauses$ terminates with success when started
in the state $\sigma$. Proofs for $\Downarrow_p$ involve finding the
matching clause for a state $\sigma$ and building a subproof for
$\Downarrow_e$ for the appropriate case of the \linl{match} statement
in the matching clause.

\subsection{Details}
\label{sec:details}

\subsubsection{Well-Formed Programs}
\label{sec:well-formed-programs}
We consider only \emph{well-formed} programs, which have
\emph{disjoint} and \emph{exhaustive} clauses. That is, for every
$M\in\Mm$ and $\sigma\in\Asp(M)$, there is precisely one clause $c$
for which $\match(\pat_c, \sigma)$ succeeds, where $\pat_c$ denotes
the state pattern for the clause $c$.

We emphasize that variables in $\lang$ programs are \emph{never} bound
to $\lang$ expressions. They are instead replaced either by components
of the syntax tree data type or the state data type. For example, a
variable $x$ might be bound to position $2$ in the word $w = abc$,
while a variable $y$ might be bound to the letter $b$. Variables are
\emph{bound} in $\lang$ programs by the state pattern at the beginning
of each clause, by the alphabet patterns in each case of a
\linl{match} statement, and in \linl{all} and \linl{any}
expressions. Well-formed programs do not have free variables.

\subsubsection{Computable Function Parameters}
\label{sec:comp-funct-param}
The functions $g\in S$ appearing in \linl{any} and \linl{all}
expressions compute finite sets, e.g., elements or sets of elements
from the domain of the structure. These elements are then bound to
variables in \linl{any} and \linl{all} expressions. For example, in a
totally-ordered structure like a word, we could use the function
$g(l,r) = [l, r]$, or a variant $g'(l,r) = [l+1,r]$ to compute sets of
consecutive positions. For convenience we allow functions $g$ to occur
in expressions that denote states. For example, if states are pairs of
word positions, we may write $(2,g(x))$, which can be evaluated to a
state once $x$ is bound. The condition in the premise of the
\textsc{Call} rule in \Cref{fig:big-step} uses a function called
$\norm$ to reduce an expression like this to a state. The functions
$g\in S$ and $f\in B$ in fact represent a family of functions, one for
each $M\in\Mm$. We write $\eval(M,g(v))$ and $\eval(M,f(v))$ to denote
the result of computing $f$ and $g$ in a structure $M$ with arguments
$v$.

\subsubsection{Negation}
\label{sec:negation}
The reader may wonder why there is no negation in the expression
syntax for $\lang$ and why the conditional for \linl{if} expressions
does not allow an arbitrary expression. These apparent restrictions
are only for simplicity. If we wanted to include negation and more
general conditionals, we could augment the states of any $\lang$
program to track the parity of the number of negations seen at any
point in a program execution; but we prefer to keep this complexity
out of the semantics. The effect of negation can easily be
accomplished by writing dual programs that implement dual operations
in particular states, as described in \Cref{sec:modal-expr-separ}.

\begin{figure}
\begin{mathpar}
  \inferrule*[left=Prog]
  {
    \tau = \match(c_i, \sigma)\\
    \tau' = \match_{c_i}(\tau(\alpha_k), n.\lab) \\
    (M, n, \clauses, \tau'(\tau(e_k))) \Downarrow_e
  }
  {
    (M, n, \clauses = \{\,c_1 \ldots c_m\,\}, \sigma) \Downarrow_p
  }

  \inferrule*[left=Call]
  {
    (M, n.c, \clauses, \sigma') \Downarrow_p \\
     \sigma' = \norm(M, \sigma(v))\\
  }
  {
    (M, n, \clauses, \hbox{\zlstinline!P($M$,\ $\sigma(v)$,\ $c$)!}) \Downarrow_e
  }

  \inferrule*[left=Bool]
  {
    \eval(M, f(v)) = \top\\
  }
  {
    (M, n, \clauses, f(v)) \Downarrow_e
  }\\

  \inferrule*[left=And]
  {
    (M, n, \clauses, e) \Downarrow_e \\
    (M, n, \clauses, e') \Downarrow_e \\
  }
  {
    (M, n, \clauses, e\ \hbox{\zlstinline!and!}\ e') \Downarrow_e
  }

  \inferrule*[left=True]
  {
  }
  {
    (M, n, \clauses, \hbox{\zlstinline!True!}) \Downarrow_e
  }\\

  \inferrule*[left=Then]
  {
    (M, n, \clauses, e_1) \Downarrow_e  \\
    \eval(M, f(v)) = \top \\
  }
  {
    (M, n, \clauses, \hbox{\zlstinline!if\ $f(v)$ then\ $e_1$ else\ $e_2$!}) \Downarrow_e
  }

  \inferrule*[left=Or1]
  {
    (M, n, \clauses, e) \Downarrow_e \\
  }
  {
    (M, n, \clauses, e\ \hbox{\zlstinline!or!}\ e') \Downarrow_e
  }\\

  \inferrule*[left=Else]
  {
    (M, n, \clauses, e_2) \Downarrow_e  \\
    \eval(M, f(v)) = \bot \\
  }
  {
    (M, n, \clauses, \hbox{\zlstinline!if\ $f(v)$ then\ $e_1$ else\ $e_2$!}) \Downarrow_e
  }

  \inferrule*[left=Or2]
  {
    (M, n, \clauses, e') \Downarrow_e \\
  }
  {
    (M, n, \clauses, e\ \hbox{\zlstinline!or!}\ e') \Downarrow_e
  }\\

  \inferrule*[left=All]
  {
    (M, n, \clauses, \{x\mapsto v_1\}(e)) \Downarrow_e \ \ \cdots\ \
    \ (M, n, \clauses, \{x\mapsto v_l\}(e)) \Downarrow_e \quad
    \eval(M, g(v)) = \{\, v_1\, \ldots\, v_l\,\}
  }
  {
    (M, n, \clauses, \hbox{\zlstinline!all (LAM$x$.\ $e$)\  $g(v)$!}) \Downarrow_e
  }\\

  \inferrule*[left=Any]
  {
    (M, n, \clauses, \{x\mapsto v_i\}(e)) \Downarrow_e \qquad v_i \in \eval(M, g(v))
  }
  {
    (M, n, \clauses, \hbox{\zlstinline!any (LAM$x$.\ $e$)\  $g(v)$!}) \Downarrow_e
  }
\end{mathpar}
\caption[LoF]{Semantics for $\lang$. If the node $n.c$ does not exist, then the \LeftTirNameStyle{Call} rule does not apply.
}
\label{fig:big-step}
\end{figure}

\newpage

\section{Translating $\lang$ Programs to Two-way Tree Automata}
\label{sec:translate}
If for every structure $M$, the states $\Asp(M)$ are computable and
finite, then a $\lang$ program can be translated to a two-way
alternating tree automaton over syntax trees. $\lang$ draws attention
to a small subset of tree automata that serve as semantic evaluators,
and which have state spaces and alphabets that can be highly
structured. Most $\lang$ programs we have considered consist of at
most a few clauses, with each clause handling the semantics for a
large set of related states.

\emph{Two-way alternating tree automata} over states $Q$ and alphabet
$\Delta$ assign to each state and symbol a positive Boolean formula
like the following
\begin{align*}
  \delta(q,a) \,\,= \,\, (q'',0) \vee (q',-1)\wedge (q,1) \wedge (q,2).
\end{align*}
This formula stipulates that in state $q$ reading symbol $a$, the
automaton can \emph{either} continue from the current position ($0$)
in state $q''$ or else it should succeed in several new positions:
from the parent ($-1$) in state $q'$ \emph{and} from the left ($1$)
and right ($2$) children, both in state $q$. More generally, the
transitions are of the following form:
\begin{align*}
  \delta(q,a) \in \Bb^{\mathtt{+}}(Q\times \{-1,0,...,k\}) \qquad
  \text{where} \,\, q\in Q, \,\, a\in \Delta,\,\, k = \arity{a},
\end{align*}
and they describe the viable next states and directions for the
machine. It can either move up ($-1$), down (numbers $> 0$), or stay
at the same node ($0$) on the input tree while changing
state. Satisfying Boolean assignments to the set
$Q\times\{-1,0,...,k\}$ describe strategies to build accepting
automaton runs along an input syntax tree, with the two components of
$Q$ and $\{-1,0,...,k\}$ corresponding to the \emph{next state} and
the \emph{direction to move}, respectively. We next explain how
programs in $\lang$ have a simple, straightforward translation to such
automata.

Given a well-formed program $P$, a structure $M$, the set $\Asp(M)$,
and a distinguished state $\sigma_i\in\Asp(M)$, we can build a two-way
alternating tree automaton $\aut(P, M)$ that accepts precisely the
syntax trees (rooted at $n$) for which
$(M, n, \clauses(P),\sigma_i) \Downarrow_p$ holds. The states of
$\aut(P, M)$ are the members of $\Asp(M)$ and the initial state is
$\sigma_i$. For each clause $c\in\clauses(P)$ of the form
\begin{center}
  \linl{$P$($M$,\ $\sigma(z)$,\ $n$) = match\ $n.\lab$ with\ $\alpha_1$ -> e_1\ $\ldots\,\, \alpha_m$ -> e_m}
\end{center}
the translation creates a set of transitions. For each matching
$\sigma\in\Asp(M)$ there is a transition for each $a\in \Delta$. We
write $\subst = \match_c(\alpha_i, a)$ to mean that $a$ matches the
alphabet pattern $\alpha_i$ in $c$ with unifying substitution $\subst$
and it does not match $\alpha_j$ for any $j < i$. We write $\tau(e)$
for the application of $\tau$ to $e$, which substitutes $\tau(x)$ for
all free occurrences of $x$ in $e$, for all $x$ in the domain of
$\tau$.

Suppose $\subst = \match(\pat_c,\sigma)$ for some clause $c$ with
cases \linl{$\alpha_i$ ->\ $e_i$}. For each symbol $a\in \Delta$, with
$\subst' = \match_c(\subst(\alpha_i), a)$, we have the transition
\begin{align*}
  \delta(\sigma, a) = \tranz(\subst'(\subst(e_i))), 
\end{align*}
with $\tranz$ described below. Any $a\in\Delta$ for which no alphabet
pattern matches, and which is therefore not covered above, is assigned
$\delta(\sigma, a) = \bot$. Thus a \linl{match} statement need not
specify what to do for alphabet symbols that, say, must never be read
in specific states of the program.

Expressions $e$ are translated into positive Boolean formulae using
the function $\tranz$ defined below.

\vspace{0.1in}
\begin{minipage}[t]{0.6\textwidth}
\centering
\begin{tabularx}{0.8\textwidth}{l l l l l l l}
  $\tranz(\text{\linl{True}})$ &$=$ & $\top$ & \quad &
  $\tranz(\text{\lstinline!$e_1$\ and\ $e_2$!})$ &$=$ & $\tranz(e_1) \wedge\tranz(e_2)$ \\
  $\tranz(\text{\linl{False}})$ &$=$ & $\bot$ & \quad &
  $\tranz(\text{\lstinline!$e_1$\ or\ $e_2$!})$ &$=$ & $\tranz(e_1) \vee \tranz(e_2)$ \\
  $\tranz(\text{\linl{$f$($v$)}})$ &$=$ & $\eval(M, f(v))$ & \quad & & &\\
\end{tabularx}
\end{minipage}

\begin{minipage}[t]{0.6\textwidth}
  \centering
\begin{tabularx}{0.8\textwidth}{l l}
  $\tranz(\text{\linl{if\ $f(v)$ then\ $e_1$ else\ $e_2$}})$ \quad &$=$\,\,\, $\begin{array}{l} \tranz(e_1)
                            \quad \text{if}\,\, \eval(M, f(v)) = \top \\
                            \tranz(e_2)
                            \quad \text{else}
                          \end{array}$ \\
  $\tranz(\text{\lstinline!all (LAM $x$.\ $e$)\ $g$($v$)!})$ \quad
  &$=$ \quad $\bigwedge_{v'\,\in\, \eval(M,\,g(v))}\tranz(\{x\mapsto v'\}(e))$\\
  $\tranz(\text{\lstinline!any (LAM $x$.\ $e$)\ $g$($v$)!})$ \quad
  &$=$ \quad $\bigvee_{v'\,\in\, \eval(M,\,g(v))}\tranz(\{x\mapsto v'\}(e))$\\
  $\tranz(\text{\linl{$P$($M$,\ $\sigma(v)$,\ $n.\dir$)}})$ \quad &$=$ \quad
                                                               $(\sigma,\dir)$
                                                               \quad
                                                               \text{where
                                                               }
                                                               $\sigma
                                                               = \norm(M,\sigma(v))$
\end{tabularx}
\end{minipage}
\vspace{0.1in}

We can now state one other well-formedness condition for $\lang$
programs. Namely, for every $M$, $\sigma\in\Asp(M)$, $a\in\Delta$, and
clause $c$ such that $\subst = \match(\pat_c,\sigma)$, if
$\subst' = \match_c(\subst(\alpha_i), a)$ for the case
\linl{$\alpha_i$ -> $\ e_i$} in $c$, then we must have:
\begin{align*}
  \tranz(\subst'(\subst(e_i))) \in \Bb^{\mathtt{+}}(\Asp\times
  \{-1,0,\ldots,\arity{a}\}).
\end{align*}
In other words, the movement of a program along the syntax tree must
respect symbol arities.

  \subsection{Example Construction}
\label{sec:example-construction}

Here we explicitly show the construction of an automaton for a
restricted regular expression language on the alphabet
$\{a,b\}$. Assume regular expressions as in \Cref{sec:regular-expressions}
but without union, intersection, and negation. Our evaluator
simplifies as follows:

\begin{lstlisting}
  Reg($w$, $(l,r)$, $n$) = match $n.\lab$ with
     $*$  -> if ($l = r$) then True else
              any (LAM$x$. Reg($w$, $(l,x)$, $n.\child_1$) and Reg($w$, $(x,r)$, $n.\stay$)) $[l+1,r]$
     $\,\cdot$  -> any (LAM$x$. Reg($w$, $(l,x)$, $n.\child_1$) and Reg($w$, $(x,r)$, $n.\child_2$)) $[l,r]$
     $x$  -> $r = l+1$ and $w(l) = x$
\end{lstlisting}

  Consider the word $w=\mathit{abb}$. The resulting two-way
  alternating automaton $\aut(\reg, w)$ is constructed as
  follows. See~\cite{tata} for definitions and results for such
  automata. The state set is
\begin{align*}
  Q=\{&(1,1), (1,2), (1,3), (1,4), (2,2), (2,3), (2,4), (3,3),
  (3,4), (4,4),\\
  &\dual(1,1), \dual(1,2), \dual(1,3), \dual(1,4), \dual(2,2),
    \dual(2,3),
    \\ &\dual(2,4), \dual(3,3),
\dual(3,4), \dual(4,4) \}.
\end{align*}

The initial state is $(1,|\mathit{abb}|+1)=(1,4)$. Here are just some
of the many transitions:

\begin{itemize}
\item[] $\delta((1,2),a) = \top$
\item[] $\delta((2,3),b) = \top$
\item[] $\delta((3,4),b) = \top$
\item[] $\delta((1,4),\cdot) = ((1,1),1) \wedge
  ((1,4),2) \vee ((1,2),1) \wedge
  ((2,4),2) \vee ((1,3),1) \wedge ((3,4),2)$
\item[] $\qquad\qquad\qquad \vee ((1,4),1) \wedge
  ((4,4),2)$
\item[] $\delta((1,3),\cdot)=\cdots$
\item[] $\delta((1,2),\cdot)=\cdots$
\item[] $\delta((1,1),\cdot)= ((1,1),1)\wedge ((1,1),2)$
\item[] $\delta((2,4),\cdot)=\cdots$
\item[] $\delta((2,3),\cdot)=\cdots$
\item[] $\delta((2,2),\cdot)=\cdots$
\item[] $\delta((3,4),\cdot)=\cdots$
\item[] $\delta((3,3),\cdot)=\cdots$
\item[] $\delta((4,4),\cdot)=\cdots$
\item[] $\delta((1,1),*) = \top$
\item[] $\delta((2,2),*) = \top$
\item[] $\delta((3,3),*) = \top$
\item[] $\delta((4,4),*) = \top$
\item[]
  $\delta((1,4),*) = ((1,2),1)\wedge((2,4),0) \vee
  ((1,3),1)\wedge((3,4,0)) \vee ((1,4),1)\wedge((4,4,0))$
\item[] $\delta((2,4),*) = \cdots$
\item[] $\delta((3,4),*) = \cdots$
\item[] $\delta((1,3),*) = \cdots$
\item[] $\delta((2,3),*) = \cdots$
\item[] $\delta((1,2),*) = \cdots$
\end{itemize}

All other transition formulas not already suggested above are
$\bot$.





\newpage

\section{Computation Tree Logic}
\label{sec:ctl}
We want a semantic evaluator for CTL syntax trees $\varphi$ over
pointed Kripke structures $G=(W,s,E,P)$ that checks whether
$G\models\varphi$. Like modal logic, the semantic aspects again
include nodes of $G$, but now also a counter to interpret path
quantifiers recursively.

We consider the following grammar for CTL formulas, from which the
other standard operators can be defined:
\begin{align*}
  \varphi\Coloneqq a\in\Sigma \grammarsep \varphi\vee\psi\grammarsep
  \neg\varphi\grammarsep \EG\varphi\grammarsep \Ex(\varphi\Until\psi)\grammarsep \EX\varphi
\end{align*}
The semantics for path quantifiers is given recursively based on the following:
\begin{align*}
  G,w\models \EX\varphi \,\,\Leftrightarrow\,\, \exists w'.\, E(w,w')\text{ and }
  G,w'\models\varphi
\end{align*}
with the other two path quantifiers interpreted according to the
equivalences
\begin{align*}
 (1)\quad \EG\varphi \equiv \varphi \wedge \EX(\EG\varphi)
  \quad\text{and}\quad (2)\quad
 \Ex(\varphi\Until\psi) \equiv \psi \vee \varphi \wedge \EX (\Ex(\varphi\Until\psi)),
\end{align*}
with $(1)$ understood as a greatest fixpoint and $(2)$ as a least
fixpoint, as we discuss shortly. The $\ctl$ program is given in
\Cref{fig:ctl}. Below, and in the program, we write \dblqt{$Ew$} as a
shorthand for $\{w'\in W\,:\, E(w,w')\}$. Fix a Kripke structure
$G=(W,s,E,P)$.

The alphabet ADT is trivial. The state ADT consists of two parts:
\begin{align*}
  \Asp(G) &\coloneqq \{w, \no(w) \,:\, w\in W\} \sqcup \Cnt(G) \\
  \Cnt(G) &\coloneqq \{(w, i),\, (\no(w), i) \,:\, w\in W,\, 0\le i\le|W|\},
\end{align*}
the first being states of the form $w$ or $\no(w)$, which do not
involve a counter, and the second being states of the form $(w,i)$ or
$(\no(w),i)$, which use a counter to verify path quantifiers.

The counter value $i$ tracks the stages of a least or greatest
fixpoint computation. If the counter is being used to verify a formula
$\EG\varphi$ holds at some $w\in W$, then this is a greatest
fixpoint. On the other hand, if the counter is being used to verify a
formula $\Ex(\varphi\Until\psi)$ does \emph{not} hold at some
$w\in W$, then it is a least fixpoint. To understand this, let us
consider the equivalences $(1)$ and $(2)$ as monotone functions over
$\Pp(W)$. We write
$\llbracket \varphi\rrbracket_G \coloneqq \{w\in W\,:\,
G,w\models\varphi\}$ for the set of states where a given formula
holds. Now observe that $(1)$ and $(2)$ correspond to the monotone
functions

\vspace{0.2in}
\begin{minipage}[t]{\textwidth}
\centering
\begin{tabularx}{0.7\linewidth}{r c c c l}
  $\EG_\varphi(X)$ & $=$ & $\llbracket\varphi\rrbracket$ & $\cap$ & $\left\{w\in W\,:\, Ew
  \cap X \neq \emptyset\right\}$ \\
  $\EU_{\varphi,\psi}(X)$ & $=$ & $\llbracket\psi\rrbracket$ &
  $\cup$ & $\left\{w\in W\,:\, Ew
  \cap X \neq \emptyset\right\} \cap \llbracket\varphi\rrbracket$,
\end{tabularx}
\end{minipage}
\vspace{0.1in}

\noindent with $\llbracket\EG\varphi\rrbracket_G = \gfp(\EG_\varphi)$ and
$\llbracket\Ex(\varphi\Until\psi)\rrbracket_G =
\lfp(\EU_{\varphi,\psi})$. The program $\ctl$ in \Cref{fig:ctl} has
the property that for all $G=(W,s,E,P)$, $w\in W$, $0\le i\le|W|$, and
$\varphi$:

\vspace{0.1in}
\begin{minipage}[t]{\linewidth}
  \centering
  \def\arraystretch{1.4}
  \begin{tabularx}{0.8\linewidth}{l r l l}
    & $(G,\route(\EG\varphi), \clauses(\ctl), (w,i))$ $\Downarrow_p$
    &
    $\Leftrightarrow$ & $w\in \EG_\varphi^{|W|-i}(W)$ \\
    and &
    $(G,\route(\Ex(\varphi\Until\psi)), \clauses(\ctl), (\no(w), i))$
    $\Downarrow_p$ & $\Leftrightarrow$ &
    $w\notin \EU_{\varphi,\psi}^{|W|-i}(\emptyset)$.
  \end{tabularx}
\end{minipage}
\vspace{0.1in}

\noindent The task of evaluating
$w\notin \llbracket\EG\varphi\rrbracket_G$ needs no counter because
there is always a finite computation witnessing non-membership for
$\EG$-formulas. The task of evaluating whether
$w\in \llbracket\Ex(\varphi\Until\psi)\rrbracket_G$ needs no counter
for the same reason. There are always proofs of \emph{removal} from a
greatest fixpoint and proofs of \emph{inclusion} in a least
fixpoint. Note we do not mind infinite loops in the other cases
because we interpret programs as tree automata with reachability
acceptance.

\begin{theorem}
  CTL separation for finite sets $\Pos$ and $\Neg$ of finite pointed
  Kripke structures, and grammar $\grammar$, is decidable in time
  $\Oo(2^{\poly(mn^2)}\cdot|\grammar|)$, where $m=|\Pos|+|\Neg|$ and
  $n=\max_{G\in \Pos\cup\Neg}|G|$.
\end{theorem}
\begin{proof}
  We have $|\Asp(G)| = \Oo(|W|^2)$, and the rest follows by
  \Cref{lemma:mainlemma} and \Cref{corollary:complexity-lemma}.
\end{proof}

\newpage
\begin{figure}[H]
  \centering
\begin{lstlisting}
  CTL($G$, $w$, $n$) =
    match $n.\lab$ with
     $\EG$ -> CTL($G$, $w$, $0$, $n.\stay$)
     $\EU$ -> CTL($G$, $w$, $n.\child_2$) or
           (CTL($G$, $w$, $n.\child_1$) and (any (LAM$z.$ CTL($G$, $z$, $n.\stay$)) $Ew$))
     $\EX$ -> any (LAM$z.$ CTL($G$, $z$, $n.\child_1$)) $Ew$
     $\,\,\vee$ -> CTL($G$, $w$, $n.\child_1$) or CTL($G$, $w$, $n.\child_2$)
     $\,\,\neg$ -> CTL($G$, $\no(w)$, $n.\child_1$)
     $\,\,x$ -> $x\in P(w)$

  CTL($G$, $\no(w)$, $n$) =
    match $n.\lab$ with
     $\EU$ -> CTL($G$, $\no(w)$, $0$, $n.\stay$)
     $\EG$ -> CTL($G$, $\no(w)$, $n.\child_1$) or (all (LAM$z.$ CTL($G$, $\no(z)$, $n.\stay$)) $Ew$)
     $\EX$ -> all (LAM$z.$ CTL($G$, $\no(z)$, $n.\child_1$)) $Ew$
     $\,\,\vee$ -> CTL($G$, $\no(w)$, $n.\child_1$) and CTL($G$, $\no(w)$, $n.\child_2$)
     $\,\,\neg$ -> CTL($G$, $w$, $n.\child_1$)
     $\,\,x$ -> $x\notin P(w)$

  CTL($G$, $w$, $i$, $n$) =
    match $n.\lab$ with
     $\EG$ -> if $i=|W|$ then True
            else CTL($G$, $w$, $n.\child_1$) and (any (LAM$z.$ CTL($G$, $z$, $i+1$, $n.\stay$)) $Ew$)

  CTL($G$, $\no(w)$, $i$, $n$) =
    match $n.\lab$ with
     $\EU$ -> if $i=|W|$ then true
            else CTL($G$, $\no(w)$, $n.\child_2$) and
                  (CTL($G$, $\no(w)$, $n.\child_1$) or
                   (all (LAM$z.$ CTL($G$, $\no(z)$, $i+1$, $n.\stay$)) $Ew$))
\end{lstlisting}
\caption{$\ctl$ evaluates CTL formulas $\varphi$ against an input
  pointed Kripke structure $G$ and checks $G\models\varphi$. }
  \label{fig:ctl}
\end{figure}

\newpage
\section{Finite-Variable First-Order Logic}
\label{sec:fo}
Fix a finite relational signature with relation symbols $R_i$ and a
set of variables $V=\{x_1,\ldots,x_k\}$. We write a program $\fo$ that
evaluates an $\FOk$ syntax tree $\varphi$ against a relational
structure $M$ and checks $M\models\varphi$. The aspects are the
(partial) assignments to variables $V$:
\begin{align*}
  \Asp(M)\coloneqq \{\gamma, \no(\gamma) \,:\, \gamma \in [V\rightharpoonup M]\},
\end{align*}
and $|\Asp(M)|=\Oo(|M|^k)$. The alphabet ADT consists of unary
constructors for $\forall$ and $\exists$, as well as $l$-ary
constructors for each $l$-ary relation symbol $R$. The program $\fo$
together with its (omitted) dual allows us to derive the main result
of~\cite{popl22}. The other results can also be derived, e.g., $\FOk$
with recursive definitions can be interpreted using a combination of
two-way navigation as in \Cref{sec:cfg} and counters as in
\Cref{sec:ctl}.

\begin{figure}[H]
  \centering
\begin{lstlisting}
  FO($M$, $\gamma$, $n$) =
    match $n.\lab$ with
      $\wedge$   -> FO($G$, $\gamma$, $n.\child_1$) and FO($G$, $\gamma$, $n.\child_2$)
      $\vee$   -> FO($G$, $\gamma$, $n.\child_1$) or FO($G$, $\gamma$, $n.\child_2$)
      $\neg$   -> FO($G$, $\no(\gamma)$, $n.\child_1$)
      $\forall x$  -> all (LAM$z$. FO($G$, $z$, $n.\child_1$)) $\{\update{\gamma}{x}{a} \,:\, a\in M\}$
      $\exists x$  -> any (LAM$z$. FO($G$, $z$, $n.\child_1$)) $\{\update{\gamma}{x}{a} \,:\, a\in M\}$
      $R(\overline{x})$ -> $\gamma(\overline{x})\in R^M$
\end{lstlisting}
\caption{$\fo$ evaluates first-order logic formulas $\varphi$ against
  an input relational structure $M$ and checks $M\models\varphi$. }
  \label{fig:fo-clause1}
\end{figure}
}{}

\end{document}